\documentclass[twoside,11pt]{article}

\usepackage{blindtext}

\usepackage[preprint]{jmlr2e}

\usepackage{bbm}
\usepackage{amsmath}
\usepackage{csquotes}
\usepackage{enumitem}
\usepackage{booktabs}
\usepackage{tabularx}

\newcommand{\RR}{\mathbbmss{R}} % Real numbers
\newcommand{\NN}{\mathbbmss{N}} % Natural numbers
\newcommand{\EE}[1]{\mathbbmss{E}\left[#1\right]} % Expected value
\newcommand{\PP}[1]{\mathbbmss{P}\left[#1\right]} % Probability measure
\newcommand{\VV}[1]{\mathbbmss{V}\left[#1\right]} % Variance
\newcommand{\one}{\mathbbmss 1} % Indicator function
\newcommand{\Normal}[2]{\mathcal{N}\left(#1,#2\right)} % Normal/Gaussian
\newcommand{\pto}{\overset{p}{\rightarrow}} % Convergence in probability
\newcommand{\dto}{\leadsto} % Convergence in distribution
\newcommand{\Lto}[1]{\overset{\mathbbmss{L}^{#1}}{\rightarrow}} % Convergence in L^#1
\newcommand{\influence}[1]{\varphi\left(#1\right)}
\newcommand{\hatinfluence}[1]{\hat{\varphi}\left(#1\right)}
\newcommand{\influencesquared}[1]{\varphi^2\left(#1\right)}
\newcommand{\data}{\mathcal{D}}
\newcommand{\projection}[2]{v_{#1}\left(#2\right)}
\newcommand{\hatprojection}[2]{\hat{v}_{#1}\left(#2\right)}
\newcommand{\richprojection}[3]{v_{#1}^{#2}\left(#3\right)}
\newcommand{\richhatprojection}[3]{\hat{v}_{#1}^{#2}\left(#3\right)}

\newcommand{\richhatinfluence}[2]{\hat{\varphi}^{#1}\left(#2\right)}
\newcommand{\indep}{\perp \!\!\! \perp}

\usepackage[usenames,dvipsnames]{xcolor}
\newcommand{\lorenzo}[1]{#1}

\newcommand{\drfos}{\hat{\beta}_{\texttt{DR-FoS}}}
\newcommand{\richdrfos}[1]{\hat{\beta}_{\texttt{DR-FoS}}^{#1}}
\newcommand{\DRFOS}{\texttt{DR-FoS}}

\usepackage{jmlr2e}

\usepackage{lastpage}
\jmlrheading{23}{2022}{1-\pageref{LastPage}}{1/21; Revised 5/22}{9/22}{21-0000}{Testa et al.}

\ShortHeadings{\DRFOS{}}{Testa et al.}
\firstpageno{1}

\begin{document}

\title{Doubly-Robust Functional Average Treatment Effect Estimation}

\author{\name Lorenzo Testa \email lorenzo@stat.cmu.edu \\
       \addr Department of Statistics \& Data Science, Carnegie Mellon University \\
       L'EMbeDS, Sant'Anna School of Advanced Studies
       \AND
       \name Tobia Boschi \email tobia.boschi@ibm.com \\
       \addr IBM Research Europe
       \AND
       \name Francesca Chiaromonte \email fxc11@psu.edu \\
       \addr Department of Statistics, Penn State University \\
       L'EMbeDS, Sant'Anna School of Advanced Studies
       \AND
       \name Edward H.~Kennedy \email edward@stat.cmu.edu \\
       \addr Department of Statistics \& Data Science, Carnegie Mellon University
       \AND
       \name Matthew Reimherr* \email mlr36@psu.edu \\
       \addr Department of Statistics, Penn State University \\
       Amazon Science
       }

\maketitle

\begin{abstract}%   <- trailing '%' for backward compatibility of .sty file
     Understanding causal relationships in the presence of complex, structured data remains a central challenge in modern statistics and science in general. While traditional causal inference methods are well-suited for scalar outcomes, many scientific applications demand tools capable of handling functional data -- outcomes observed as functions over continuous domains such as time or space. Motivated by this need, we propose \DRFOS{}, a novel method for estimating the Functional Average Treatment Effect (FATE) in observational studies with functional outcomes. \DRFOS{} exhibits double robustness properties, ensuring consistent estimation of FATE even if either the outcome or the treatment assignment model is misspecified. By leveraging recent advances in functional data analysis and causal inference, we establish the asymptotic properties of the estimator, proving its convergence to a Gaussian process. This guarantees valid inference with simultaneous confidence bands across the entire functional domain. Through extensive simulations, we show that \DRFOS{} achieves robust performance under a wide range of model specifications. Finally, we illustrate the utility of \DRFOS{} in a real-world application, analyzing functional outcomes to uncover meaningful causal insights in the SHARE ({\em Survey of Health, Aging and Retirement in Europe}) dataset.
\end{abstract}


\section{Introduction}

Causal inference has become a foundational tool for understanding the effects of interventions in a variety of fields, including medicine \citep{prosperi2020causal}, economics \citep{varian2016causal}, and the social sciences \citep{imbens2024causal}. Most causal inference methods are developed for settings in which the outcome of interest is a scalar variable, such as a single measurement of health status or income level. However, there is a growing need to expand these methods to accommodate more complex data structures, particularly functional data, where the outcome is observed as a function over a continuous domain, such as time or space \citep{kokoszka2017introduction, ramsay2005}. Functional outcomes arise in numerous applications, such as longitudinal studies where medical or wearable devices are used to record patients' biometrics over time \citep{boschi2024new, jeong2024identifying,lundborg2022conditional}, epidemiological studies where infectious diseases are tracked across spatiotemporal domains \citep{boschi2021functional, boschi2023contrasting}, and neuroscience studies where brain activity is monitored over time and space \citep{boschi2024fasten, qi2018function}. Functional data introduces specific challenges to causal inference, necessitating theory and methodology that can capture the infinite-dimensional nature of these outcomes while providing reliable estimates of treatment effects.

Despite the scientific relevance of this problem, the estimation of treatment effects in functional data settings has not received much attention. The primary challenge lies in handling the complexity of functional data, which requires specialized tools to account for the structure and dependencies across the continuous domain. Traditional scalar methods for estimating treatment effects cannot be applied directly, and current functional data techniques either are designed for settings without access to additional covariates \citep{cremona2018iwtomics} or model outcomes with non-robust tools such as function-on-scalar regression \citep{ecker2024causal, ieva2025enhancing}, which may perform poorly if the model is misspecified. A relevant exception is \citet{liu2024double} -- where the authors provide, under strong parametric assumptions, an estimator that can be robust to some forms of model misspecification. 

In this paper, we introduce the Doubly-Robust Function-on-Scalar estimator (\DRFOS{}), a novel approach for estimating the Functional Average Treatment Effect (FATE) in settings with functional outcomes. Our proposed method draws from the concept of double robustness, a property in causal inference that provides consistent estimation even if only one of two models -- either the outcome model or the treatment assignment model -- is correctly specified. Double robustness has been widely used in scalar treatment effect estimation, particularly through methods such as Augmented Inverse Probability Weighting (AIPW) and Targeted Maximum Likelihood Estimation \citep{tsiatis2006semiparametric, van2006targeted}. The extension of double robustness to functional data requires a careful analysis of the underlying functional structure and the choice of estimation strategies that can handle high-dimensional functional representations.
\DRFOS{} leverages recent advances in functional data analysis, combining ideas from functional central limit theorems \citep{dette2020functional} and simultaneous confidence bands \citep{liebl2023fast, pini2017interval} with doubly-robust estimation techniques to provide a flexible, theoretically sound estimator for FATE.

The contributions of this paper, building upon prior work in causal inference \citep{kennedy2019estimating, lin2023causal} and functional data analysis \citep{dette2020functional}, are threefold. First, we develop the \DRFOS{} estimator, which integrates information from both the outcome model and the treatment assignment model, and features double robustness properties to mitigate the risk of misspecification in either. Second, we rigorously analyze the asymptotic properties of the \DRFOS{} estimator, showing that it converges to a Gaussian process under weak regularity conditions. We also provide, under minimal assumptions, estimation procedures that facilitate the construction of confidence bands with simultaneous guarantees, thus enabling inference on the estimated functional treatment effects. Finally, as intermediate steps in our proofs, we establish novel results on the asymptotic distribution of the AIPW estimator for the average treatment effect (ATE) in the case of multivariate vector outcomes, and on asymptotic properties of the inverse probability weighting (IPW) estimator for functional outcomes. %Additionally, we derive asymptotic properties of an inverse probability weighting (IPW) estimator for functional outcomes -- results that, perhaps surprisingly, we could not find in previous literature.

To validate the effectiveness of \DRFOS{}, we conduct a series of simulations that demonstrate its robustness and accuracy under various model misspecification scenarios. Our simulations illustrate that \DRFOS{} consistently outperforms methods based solely on the estimation of propensity score or outcome regression models, especially under misspecification, underscoring the practical utility of double robustness in functional data settings. Finally, we apply \DRFOS{} to the SHARE dataset ({\em Survey of Health, Aging and Retirement in Europe}; \citet{mannheim2005survey})
analyzing the causal effect of chronic conditions on functional indicators of quality of life.

The remainder of this paper is organized as follows. Section \ref{sec:identify} defines the target and gives causal assumptions under which the target is identifiable. Section \ref{sec:estimate} introduces the estimation procedure and the supporting theoretical results required for inference. Section \ref{sec:simulate} describes our simulation study, which provides further support for our approach. We then apply \DRFOS{} to the SHARE dataset in Section \ref{sec:apply}. Finally, Section \ref{sec:conclude} contains some concluding remarks. All code for reproducing our analysis is available at \url{https://github.com/testalorenzo/DR-FoS}.

\subsection{Notation and setup}
Let $\mathcal T = [t_1,t_2]$ be a closed bounded interval. Without loss of generality, one can consider $\mathcal T = [0,1]$. 
We observe a collection of $n$ independent and identically distributed samples, $\left\{\data_i = (A_i, X_i, \mathcal{Y}_i)\right\}_{i=1}^n$, where $A_i\in\{0,1\}$ is a binary variable that indicates whether subject $i$ belongs to the treated group ($A_i=1$) or to the control group ($A_i=0$); $X_i\in\RR^p$ is a $p$-dimensional vector of covariates; and $\mathcal{Y}_i\in C(\mathcal{T})$ is a continuous function representing the outcome defined over $\mathcal{T}$. We assume $C(\mathcal{T})$, the space of continuous functions over $\mathcal{T}$, is equipped with the uniform topology, which induces the sup-norm defined by $\|f\|=\sup_{t\in\mathcal{T}} |f(t)|$, for $f\in C(\mathcal{T})$, thus making $(C(\mathcal{T}), \|\cdot\|)$ a Banach space. We denote \textit{weak convergence} of a random object $W_n$ (which can be either finite or infinite dimensional) to a limit $W$ as $W_n\dto W$. Similarly, we denote \textit{convergence in $\mathbbmss{L}^2$ norm} as $W_n\Lto{2}W$. Finally, throughout this paper, we let $\mathcal{D} = (A,X,\mathcal{Y})$ denote an independent copy of $\mathcal{D}_i = (A_i,X_i,\mathcal{Y}_i)$.

We would like to stress that our choice of the sup-norm, following \citet{dette2020functional}, significantly differentiates our work from most of the standard literature on functional data \citep{bosq2000linear, ferraty2006nonparametric, kokoszka2017introduction, ramsay2005} and causal inference \citep{liu2024double, luedtke2024one}, which instead simplifies the analysis by assuming that functional outcomes exist in a Hilbert space (e.g., the space of square integrable functions $\mathbbmss{L}^2$). In a Banach space, standard asymptotics breaks down without further assumptions; we thus need to introduce tailored assumptions to be able to prove central limit theorems in this setting. 
Even more importantly, the inferential problem of constructing confidence bands would be ill-posed in the $\mathbbmss{L}^2$ space. In fact, the $\mathbbmss{L}^2$ norm provides a measure of the global, average deviation of a function along its domain, but does not account for pointwise or localized deviations, which are essential for constructing simultaneous confidence bands. Instead, confidence bands require control over the supremum norm, making $C(\mathcal{T})$, the space of continuous functions endowed with the sup-norm, the most natural choice for this purpose.

\section{Definition and identification}
\label{sec:identify}

Our target is the \textit{functional average treatment effect} (FATE), defined as
\begin{equation}
\label{eq:FATE}
    \beta = \EE{\mathcal{Y}^{(1)}- \mathcal{Y}^{(0)}}\,,
\end{equation}
where the expectation is taken over the data generating process. Recall $\mathcal{Y}$ is a function, so $\beta$ is as well; we omit function arguments for notational simplicity. In terms of functional analysis, if $\mathcal{P}$ is the class of probability measures on $C(\mathcal{T})$, we can represent our target as a functional $\psi:\mathcal{P}\to C(\mathcal{T})$ such that $\psi(\mathbbmss{P})=\beta$.

Despite the increased complexity of the outcome variable and target, conditions under which the latter is identifiable are standard. In particular, we leverage the following setup.

\begin{assumption}[Identifiability]
\label{ass:identify}
Let the following identifiability assumptions hold:
    \begin{enumerate}[label=\textbf{\alph*.}]
    \item \textbf{Consistency.} The potential outcome of a treatment is the same regardless of the mechanism by which the treatment is administered; that is, $\mathcal{Y} = \mathcal{Y}^{(a)}$ if $A=a$. Equivalently, $\mathcal{Y} = A\mathcal{Y}^{(1)} + (1-A)\mathcal{Y}^{(0)}$.
    \item \textbf{No unmeasured confounding.} 
    $(\mathcal{Y}^{(0)}, \mathcal{Y}^{(1)}) \indep A \,|\, X$.
    \item \textbf{Positivity.} $0<\PP{A=1\,|\,X}<1$ almost surely.
\end{enumerate}
\end{assumption}
With these assumptions in place, we are ready to introduce the proposed estimator.

\section{Estimation and inference}
\label{sec:estimate}

\subsection{Estimation}
\label{subsec:estimate}
Define the \textit{propensity score} and \textit{regression function}, respectively, as
\begin{equation}
    \pi^{(a)}(x) = \PP{A = a\,|\,X=x}=\EE{\one\{A = a\}\,|\,X = x}\,,
\end{equation}
\begin{equation}
    \mu^{(a)}(x) = \EE{\mathcal{Y}^{(a)}\,|\,X=x,\,A=a}\,,
\end{equation}
and denote their estimates by $\hat{\pi}^{(a)}$ and $\hat{\mu}^{(a)}$. Notice that the techniques or algorithms employed to produce such estimates are arbitrary; often, logistic regression or random forest classifiers are used to compute $\hat\pi^{(a)}$, while complex nonparametric tools, such as neural networks, can be employed in the estimation of $\hat{\mu}^{(a)}$.
Finally, define the \textit{case-corrected regression function} as
\begin{equation}
\begin{split}
    \gamma^{(a)}(\data) &= \mu^{(a)}(X) + \frac{\one\{A=a\}(\mathcal{Y}^{(a)} - \mu^{(a)}(X))}{\pi^{(a)}(X)} = \begin{cases}
        \mu^{(a)}(X) + \frac{\mathcal{Y}^{(a)} - \mu^{(a)}(X)}{\pi^{(a)}(X)} & \text{if $A = a$}\\
        \mu^{(a)}(X) & \text{if $A \neq a$}\,.\\
    \end{cases}
\end{split}
\end{equation}

The following Lemma is pivotal in the derivation of the \DRFOS{} estimator.
\begin{lemma}
\label{th:equivalence}
Under Assumptions \ref{ass:identify} (identifiability), the FATE defined in Eq.~\ref{eq:FATE} can be rewritten as:
\begin{equation}
    \beta = \EE{\gamma^{(1)}(\data)- \gamma^{(0)}(\data)}\,.
\end{equation}
\end{lemma}
We defer the proof of this and subsequent claims to Appendix Section~\ref{suppsec:proofs}. 

Let $\hat{\mathbbmss{P}}$ denote the sample distribution on which $\hat{\mu}^{(a)}$ and $\hat{\pi}^{(a)}$ are produced. We assume to have access to a separate sample distribution $\mathbbmss{P}_n$ independent of $\hat{\mathbbmss{P}}$, and define the \textit{one-step} functional augmented inverse probability weighting estimator $\drfos$ as
\begin{equation}
    \label{eq:drfos}
\begin{split}
    \drfos &= \mathbbmss{P}_n \left[ \hat{\gamma}^{(1)}(\data) -\hat{\gamma}^{(0)}(\data) \right] \\
    &= \frac{1}{n} \sum_{i=1}^n \left( \hat{\gamma}^{(1)}(\data_i) -\hat{\gamma}^{(0)}(\data_i) \right) \\
    &=\frac{1}{n} \sum_{i=1}^n \left[\left( \hat{\mu}^{(1)}(X_i) + \frac{A_i(\mathcal{Y}_i - \hat{\mu}^{(1)}(X_i))}{\hat{\pi}^{(1)}(X_i)} \right) -  \left( \hat{\mu}^{(0)}(X_i) + \frac{(1-A_i) (\mathcal{Y}_i - \hat{\mu}^{(0)}(X_i))}{1 - \hat{\pi}^{(1)}(X_i)} \right)\right]\,.
\end{split}
\end{equation}
Notice that if $\hat{\pi}^{(a)} = \pi^{(a)}$ and $\hat{\mu}^{(a)} \neq \mu^{(a)}$, we have $\EE{\drfos} = \beta$. Similarly, if $\hat{\pi}^{(a)} \neq \pi^{(a)}$ and $\hat{\mu}^{(a)} = \mu^{(a)}$, we again have $\EE{\drfos} = \beta$. In other words, $\drfos$ remains unbiased as long as either the propensity score model $\hat{\pi}^{(a)}$ or the outcome regression model $\hat{\mu}^{(a)}$ is correctly specified. This illustrates the \textit{double robustness} of the estimator, a property that we will further explore under additional assumptions on convergence rates later in the paper.
% \begin{remark}
%     Notice that if $\hat\pi^{(a)}=\pi^{(a)}$ and $\hat\mu^{(a)}\neq\mu^{(a)}$, we have $\EE{\drfos}=\beta$. Similarly, if $\hat\pi^{(a)}\neq\pi^{(a)}$ and $\hat\mu^{(a)}=\mu^{(a)}$, we again have $\EE{\drfos}=\beta$. In other words, as long as either $\hat\pi^{(a)}$ or $\hat\mu^{(a)}$ is unbiased, or well-specified, $\drfos$ will be unbiased as well. This property is sometimes referred to as \textit{weak} double robustness -- as it is, in a way, a qualitative property. Later on we will show that, under additional assumptions on convergence rates, \textit{strong} double robustness (a stricter, quantitative property) also holds. \textcolor{red}{[EK: I would take out this remark and just say the double robustness of the estimator will be shown later. I'm not sure this weak/strong notion is standard/useful?]}
% \end{remark}

\begin{remark}
    The derivation of $\drfos$ exploits the \textit{influence functions} of the target of inference $\beta$. Unlike in standard semiparametric theory \citep{tsiatis2006semiparametric}, the influence functions here are infinite-dimensional and take the form
    \begin{equation}
        \influence{\data} = \gamma^{(1)}(\data) - \gamma^{(0)}(\data) - \beta \,.
    \end{equation}
    However, perhaps surprisingly, in our developments we never need these infinite-dimensional functions to be well-defined. Instead, we only require that their projections -- which are finite-dimensional random vectors -- satisfy some consistency conditions. We believe that the study of infinite-dimensional influence functions deserves a focus on its own, but it is outside the scope of this paper.
\end{remark}

\begin{remark}
    The framework developed so far can be leveraged to enable \DRFOS{} to attain further robustness. In fact, multiple base models can be fitted for both the propensity score and the outcome regression, with their predictions subsequently combined using techniques such as stacking \citep{pavlyshenko2018using} or SuperLearner \citep{van2007super}. Then, these ensemble models can be incorporated into Eq.~\ref{eq:drfos}. By adopting this approach, the \DRFOS{} procedure remains consistent as long as at least one among the base models is consistent, providing further protection against model misspecification. % \textcolor{red}{[EK: this is not usually called multiple robustness]}
\end{remark}

A separate independent sample $\mathbbmss{P}_n$ can be obtained by splitting the data. To obtain full sample size efficiency, we exploit \textit{cross-fitting}. Cross-fitting works as follows. We first randomly split the observations $\{\data_1,\dots,\data_n\}$ into $J$ disjoint folds (without loss of generality, we assume that the number of observations $n$ is divisible by $J$). For each $j=1,\ldots, J$ we form $\hat{\mathbbmss{P}}^{[-j]}$ with all but the $j$-th fold, and $\mathbbmss{P}_n^{[j]}$ with the $j$-th fold.
Then, we learn $\hat{\mu}^{(a)[-j]}$ and $\hat{\pi}^{(a)[-j]}$ on $\hat{\mathbbmss{P}}^{[-j]}$, and evaluate $\drfos^{[j]}$ on $\mathbbmss{P}_n^{[j]}$. Finally, we average to obtain our \textit{cross-fitted} estimator as
\begin{equation}
    \drfos = \frac{1}{J} \sum_{j=1}^J \drfos^{[j]}\,.
\end{equation}

\begin{remark}
\label{rem:missing}
    The causal inference problem we are considering naturally resembles the missing data setup with two levels of missingness \citep{tsiatis2006semiparametric}. Let $R$ denote the missingness indicator. In this case, the prototypical observed data can be represented as $\data = (X, R, R\mathcal{Y})$. The term $R\mathcal{Y}$ highlights that we only observe $\mathcal{Y}$ when $R=1$. Under the assumptions of missing at random (MAR), that is $\mathcal{Y} \indep R \mid X$, and positivity, i.e.~$0<\PP{R=1\mid X=x}<1$ for every $x$ almost surely, the target $\beta = \EE{\mathcal{Y}}$ can be identified, and then efficiently estimated using $\drfos = n^{-1} \sum_{i=1}^n \left( \hat{\mu}(X_i) + \frac{A_i(\mathcal{Y}_i - \hat{\mu}(X_i))}{\hat{\pi}(X_i)} \right)$, where $\hat\pi$ is an estimator of $\PP{R=1\mid X=x}$ and $\hat\mu$ is an estimator of $\EE{\mathcal Y\mid X, R=1}$.
\end{remark}

\subsection{Inference}
In the well-known scalar response setting, standard semiparametric theory shows that the ATE estimator, under some conditions, is asymptotically normal \citep{kennedy2022semiparametric}. Therefore, we can show that even when the outcome is functional, \textit{pointwise asymptotic normality} holds if similar conditions are satisfied (see Corollary \ref{cor:pant} below). However, this does not imply simultaneous guarantees over the continuous domain of our functional outcome. Thus, we rely on convergence of probability measures in the space of continuous functions to show that our estimator possesses an asymptotic Gaussian process behaviour. Then, we exploit this fact to build confidence bands that guarantee simultaneous coverage, both based on the parametric bootstrap approach proposed by \citet{pini2017interval} and on the critical value functions approach proposed by \citet{liebl2023fast}.

Compared to the classical causal inference framework where the outcome is real-valued, the problem we are tackling here requires additional structure. In particular, we need to introduce more notation and assumptions. First, for any $t\in\mathcal{T}$, we denote the one-dimensional projection of $\beta$ and $\drfos$ at $t$ as $\beta(t)$ and $\drfos(t)$, respectively. Similarly, we denote the one-dimensional projection of $\influence{\data}$ and $\hatinfluence{\data}$ at $t$ as $\influence{\data;t}$ and $\hatinfluence{\data;t}$.
%First, for any $k\in\NN$ and $t_1,\dots,t_k\in\mathcal{T}$, we denote the $k$-dimensional projection vectors of $\beta$ and $\drfos$ at $(t_1,\dots,t_k)$ as $B_{t_1, \dots, t_k} = \left(\beta(t_1),\dots,\beta(t_k)\right)^T$ and $\hat{B}_{t_1, \dots, t_k} = \left(\drfos(t_1),\dots,\drfos(t_k)\right)^T$, respectively. Similarly, we denote the $k$-dimensional projection vectors of $\influence{\data_i}$ and $\hatinfluence{\data_i}$ at $(t_1,\dots,t_k)$ as $\projection{t_1,\dots,t_k}{\data_i} = (\influence{\data_i;t_1}, \dots, \influence{\data_i;t_k})^T$ and $\hatprojection{t_1,\dots,t_k}{\data_i} = (\hatinfluence{\data_i;t_1}, \dots, \hatinfluence{\data_i;t_k})^T$. 
Next, we consider the following assumptions.

\begin{assumption}[Inference]
\label{ass:estimate}
    Let the number of cross-fitting folds be fixed at $J$, and assume that: 
    \begin{enumerate}[label=\textbf{\alph*.}]
        %\item For each $j\in\{1,\dots,J\}$, and for every one-dimensional projection $\projection{t}{\data_i}$, one has $$\EE{(\richhatprojection{t}{[-j]}{\data_i} - \projection{t}{\data_i})^2} = o(1)\,.$$
        \item For each $j\in\{1,\dots,J\}$, and for every one-dimensional projection $\influence{\data;t}$, one has $$ \richhatinfluence{[-j]}{\data;t} \Lto{2} \influence{\data;t}\,.$$
        %\textcolor{red}{[EK: the LHS is not random]}
         % \item Within each cross-fitting fold, and for any $k$ dimensional finite projection $\projection{t_1,\dots,t_k}{\data_i}$, it holds that $\EE{(\richhatprojection{t_1,\dots,t_k}{q}{\data_i} - \richprojection{t_1,\dots,t_k}{q}{\data_i})^2} = o_\mathbbmss{P}(1)$ for each component $q\in\{1,\dots,k\}$;
        \item For every one-dimensional projection %$\projection{t}{\data_i}$
        $\influence{\data;t}$, one has $$
        \sqrt{n}
        \sum_{j=1}^J R_2^{[j]} = o_\mathbbmss{P}(1) \,,
        $$ % I put a fraction where it was not needed
        where %R_2^{[j]} = \hat{B}_{t}^{[-j]} - B_t + \int \richhatprojection{t}{[-j]}{\data_i}\,d\mathbbmss{P}$.
        $R_2^{[j]} = \richdrfos{[-j]}(t) - \beta(t) + \int \richhatinfluence{[-j]}{\data;t}\,d\mathbbmss{P}\,.$
        % \item For any $k$ dimensional finite projection $\projection{t_1,\dots,t_k}{\data_i}$, it holds that $n^{-1/2} \sum_{j=1}^J R_2^{[j]} = o_\mathbbmss{P}(1)$, where $R_2^{[j]} = \hat{B}_{t_1,\dots,t_k} - B_{t_1,\dots,t_k} + \int \hatprojection{t_1,\dots,t_k}{\data_i}\,d\mathbbmss{P}\,;$
        \item Given $\xi>0$, $\hat{\pi}^{(a)}$ is bounded away from $\xi$ and $1-\xi$ with probability 1.
        %\textcolor{red}{[EK: probably need bounded away from some $\epsilon>0$ and $1-\epsilon$]}
        \item For any $\delta>0$, the functional outcome satisfies
        \begin{equation}
            \EE{\sup_{|s-t|\leq \delta} \left|\mathcal{Y}(s)-\mathcal{Y}(t)\right|} \leq L \delta^\tau\,
        \end{equation}
        for some constants $L\geq0$ and $\tau>0$.
        %;
        \item For any $\delta>0$, and for $a\in\{0,1\}$, the \textit{estimated} regression function satisfies
        \begin{equation}
            \EE{\sup_{|s-t|\leq \delta} \left|\hat{\mu}^{(a)}(s)-\hat{\mu}^{(a)}(t)\right|} \leq L^{(a)} \delta^\tau\,
        \end{equation}
        for some constants $L^{(a)}\geq0$ and $\tau>0$.
    \end{enumerate}
\end{assumption}

\begin{remark}
    Taking the number of folds as fixed to rule out undesired asymptotic behaviors of cross-fitting, as well as Assumptions \textbf{a}, \textbf{b} and \textbf{c}, are standard in causal inference -- see \citet{kennedy2022semiparametric}. In particular, Assumptions \textbf{a}, \textbf{b}, and \textbf{c}
    are required to control the empirical process and the remainder term arising from the Von Mises expansion in Lemma~\ref{lem:pan}. %Notice that we require conditions on neither the influence functions $\influence{\data_i}$, nor their $k$-dimensional projections $\projection{t_1,\dots,t_k}{\data_i}$. We only need first-order assumptions on the consistency of influence function estimators $\hatprojection{t}{\data_i}$, as is standard in the literature with scalar outcomes. 
    Notice that we do not require any condition on the influence functions $\influence{\data}$. We only need first-order assumptions on the consistency of influence function estimators $\hatinfluence{\data;t}$, as is standard in the literature with scalar outcomes. 
    See Remark~\ref{rem:remainder} below for additional comments on assumption \textbf{b}.
    Assumption \textbf{d} is much lighter than what is usually required in the functional data literature. For example, the assumption of \citet{dette2020functional} on Lipschitz sample paths, i.e.~$\left|\mathcal{Y}(s)-\mathcal{Y}(t)\right|\leq L'|s-t|$, clearly implies ours with $L'=L$ and $\tau=1$. Moreover, we stress the fact that we do not require the functional outcome to be in a Hilbert space.  
    Finally, Assumption \textbf{e} is a regularity condition on the regression function estimator. This can be superfluous depending on the methodology at hand. See Example~\ref{example:ols} below, where assumption \textbf{d} implies \textbf{e}. 
\end{remark}

\begin{remark}
\label{rem:remainder}
    The remainder term in Assumption \textbf{b} can be bounded by the product of the norms of the distance between true and estimated nuisance functions. %(see Supplementary Material Section~\ref{suppsec:remainder} for a proof). 
    In fact, by direct evaluation of $R_2$, we can show that 
    \begin{equation}
    \begin{split}
        R_2 &= \int \left|\frac{1}{\pi^{(1)}(X)} - \frac{1}{\hat{\pi}^{(1)}(X)} \right|\, \left| \mu^{(1)}(X;t) - \hat{\mu}^{(1)}(X;t) \right| \pi^{(1)}(X)\,d\mathbbmss{P} \\
        &- \int \left|\frac{1}{1 - \pi^{(1)}(X)} - \frac{1}{1 -\hat{\pi}^{(1)}(X)} \right|\, \left| \mu^{(0)}(X;t) - \hat{\mu}^{(0)}(X;t) \right| (1-\pi^{(1)}(X))\,d\mathbbmss{P}\,.
    \end{split}
    \end{equation}
    Given that, by Assumption \textbf{c}, $\hat{\pi}^{(1)}(X) \geq \varepsilon$ with probability 1, by Cauchy-Schwarz we have
    \begin{equation}
        \begin{split}
            |R_2| &\leq \frac{1}{\varepsilon} \int \left|\pi^{(1)}(X) - \hat{\pi}^{(1)}(X) \right|\, \left| \mu^{(1)}(X;t) - \hat{\mu}^{(1)}(X;t) \right| \,d\mathbbmss{P} \\
            &- \frac{1}{\varepsilon}\int \left|\hat{\pi}^{(1)}(X) - \pi^{(1)}(X) \right|\, \left| \mu^{(0)}(X;t) - \hat{\mu}^{(0)}(X;t) \right| \,d\mathbbmss{P} \\
            &\leq \frac{1}{\varepsilon}\EE{\pi^{(1)}(X) - \hat{\pi}^{(1)}(X)}\left(\EE{\mu^{(1)}(X;t) - \hat{\mu}^{(1)}(X;t)} - \EE{\mu^{(0)}(X;t) - \hat{\mu}^{(0)}(X;t)
            } \right)\,.
        \end{split}
    \end{equation}
    
    The fact that the remainder is bounded by the product of the norms of the distance between true and estimated quantities is also referred to as \textit{mixed bias}, \textit{product bias}, or \textit{strong} double-robustness property \citep{wager2024causal}, as opposed to the \textit{weak} double-robustness property presented in the main text (see Section~\ref{subsec:estimate}). This property implies that the rate at which the remainder converges to 0 is given by the product of the rates at which the estimators converge to truth. This turns out to be particularly useful in our context. In fact, in several applications it would be reasonable to assume that the estimation of the regression functions $\mu^{(a)}(X)$, $a\in\{0,1\}$ is a more complex task than the estimation of $\pi^{(1)}(X)$. Indeed, $\mu^{(a)}$ maps a vector of size $p$ into a function in $C(\mathcal{T})$; $\pi^{(a)}$ maps the same vector into a variable supported in $(0,1)$. If the rate of convergence of $\hat{\pi}^{(1)}$ were parametric (i.e.~$n^{-1/2}$), then \textit{any} rate of convergence for $\hat{\mu}^{(a)}$ would suffice to guarantee the convergence of \DRFOS{} to a Gaussian process. See \citet{kennedy2022semiparametric} for additional comments on this problem.
\end{remark}

\begin{example}[Assumption \textbf{e} can be superfluous]
\label{example:ols}
    Let $\mathbb{Y} = (\mathcal{Y}_1, \dots, \mathcal{Y}_n)^T$ and, similarly, $\mathbb{X} = (X_1,\dots,X_n)^T$. The standard function-on-scalar OLS estimator is given by 
    \begin{equation}
        \hat{\theta}_{\texttt{OLS}} = (\mathbb{X}^T\mathbb{X})^{-1} \mathbb{X}^T \mathbb{Y}\,.
    \end{equation}
    The estimated regression function is thus
    \begin{equation}
        \hat{\mu}(X) = X\hat{\theta}_{\texttt{OLS}}\,.
    \end{equation}
    Clearly, with a slight abuse of notation, we have that
    \begin{equation}
    \begin{split}
        \EE{\sup_{|s-t|\leq \delta} \left|\hat{\mu}(s)-\hat{\mu}(t)\right|} &= \EE{X \EE{\sup_{|s-t|\leq \delta} \left| \hat{\theta}(s)-\hat{\theta}(t)\right| \,\Big|\,X}} \\
        &=\EE{X(\mathbb{X}^T\mathbb{X})^{-1}\mathbb{X}^T \EE{\sup_{|s-t|\leq \delta} \left| \mathbb{Y}(s)- \mathbb{Y}(t)\right|}} \\
        &\leq L'' |s-t|^\tau\,,
    \end{split}
    \end{equation}
     where $L'' = L\EE{X(\mathbb{X}^T\mathbb{X})^{-1}\mathbb{X}^T}$.
    An extension to linear smoothers goes through a similar line of reasoning. \lorenzo{We stress that Assumption~\ref{ass:estimate}\textbf{d} imposes regularity on the outcome process $\mathcal{Y}$ itself, independent of the estimation procedure. In contrast, Assumption~\ref{ass:estimate}\textbf{e} is a condition on the estimated regression functions $\hat{\mu}^{(a)}$. In practice, methods like linear smoothers will ensure that $\hat{\mu}^{(a)}$ inherits sufficient smoothness from the estimation method itself.}
\end{example}

\begin{lemma}[Asymptotic normality of finite dimensional projections]
\label{lem:pan}
    Let $k\in\NN$ and $t_1,\dots,t_k\in\mathcal{T}$ be fixed. 
    Define $\hat{B}_{t_1,\dots,t_k} = \left(\drfos(t_1),\dots,\drfos(t_k)\right)^T$ and $B_{t_1,\dots,t_k} = \left(\beta(t_1),\dots,\beta(t_k) \right)^T$. Under Assumptions~\ref{ass:identify} (identifiability) and \ref{ass:estimate} (inference), one has
    \begin{equation}
        %\sqrt{n} \left( \hat{B}_{t_1,\dots,t_k} - B_{t_1,\dots,t_k} \right) \dto \Normal{0}{\Sigma_{t_1,\dots,t_k}}\,,
        \sqrt{n} \left( \hat{B}_{t_1,\dots,t_k} - B_{t_1,\dots,t_k} \right) \dto \Normal{0}{\Sigma_{t_1,\dots,t_k}}\,,
    \end{equation}
    where $\projection{t_1,\dots,t_k}{\data} = (\influence{\data;t_1}, \dots, \influence{\data;t_k})^T$ and $\Sigma_{t_1,\dots,t_k} = \EE{ \projection{t_1,\dots,t_k}{\data} \projection{t_1,\dots,t_k}{\data}^T}\,$.
\end{lemma}

The claim in Lemma~\ref{lem:pan} is similar to the one in Theorem 5.31 of \citet{van2000asymptotic}, with the notable difference that, by exploiting cross-fitting, we greatly reduce the Donsker-type assumptions on population and estimated quantities. Our claim also resembles Lemma C.8 in \citet{martinez2023efficient}. A similar, more general, result is also found in Lemma 3 of \citet{kennedy2023semiparametric}. As stated in the following Corollary, Lemma \ref{lem:pan} also guarantees pointwise asymptotic normality.

\begin{corollary}[Pointwise Asymptotic normality]
\label{cor:pant}
    Let $t\in\mathcal{T}$ be fixed. Under the assumptions of Lemma \ref{lem:pan}, one has
    \begin{equation}
        \sqrt{n} \left( \drfos(t) - \beta(t) \right) \dto \Normal{0}{\Sigma_t}\,,
    \end{equation}
    where $\Sigma_t = \VV{\influence{\data;t}} = \EE{\influencesquared{\data;t}}\,$.
\end{corollary}

We are now ready to present the main result of this Section, which will be pivotal in the construction of simultaneous confidence bands required to perform inference. It informally states that the asymptotic distribution of $\drfos$ is a Gaussian process. 

\begin{theorem}[Convergence to Gaussian process]
\label{th:gp}
    Under Assumptions \ref{ass:identify} (identifiabiliy) and \ref{ass:estimate} (inference), one has
    \begin{equation}
        \sqrt{n}\left(\drfos - \beta \right) \dto \mathcal{GP}(0, \Sigma)\,,
    \end{equation}
    where $\Sigma(s,t) = \EE{\varphi(\data;s)\varphi(\data;t)}\,$. 
\end{theorem}

\begin{remark}
    Our CLT can be extended using the tools introduced in \citet{dette2020functional} to accommodate dependent and non-identically distributed data. However, 
    %notice that if we restrict
    we note that restricting ourselves to the i.i.d.~case, as presented here, 
    %we require 
    allows us to employ a less stringent condition on the continuity of sample paths. Also, 
    %given that we want 
    since our objective is to set the stage for 
    %the extension of 
    extending standard causal inference approaches 
    %to causal inference, 
    (which usually 
    %assumes 
    assume i.i.d.~data)
    %, 
    to functional data analysis,
    %assuming i.i.d.~data 
    considering the i.i.d.~case seems fully justified.
\end{remark}

\begin{remark}
    Inference for the functional average treatment effect using \DRFOS{}, by exploiting influence functions and cross-splitting, requires very mild assumptions compared to the proposal of \citet{liu2024double}. In fact, we make virtually no assumption about the shape and complexity of the population regression function and propensity score. Instead, \citet{liu2024double} impose strong modeling constraints on the regression function (which is assumed to be linear) and on the propensity score (which is assumed to be logistic). They also require the usual regularity conditions for asymptotic normality of M-estimators. Moreover, in our set of assumptions for inference, we only require that functional outcomes satisfy an expected Hölder condition, while \citet{liu2024double} require the much more stringent assumption that outcomes be defined in a Hilbert space, where standard asymptotic theory is applicable. 
\end{remark}

Endowed with the previous result, we can now build confidence bands for inference. An elegant pivotal approach, borrowed by \citet{liebl2023fast}, requires the estimation of a critical value function $\hat{u}^*_{\alpha/2}$, which guarantees control over the false positive rate. We refer the interested reader to \citet{liebl2023fast} for additional details and note that, while this approach possesses a variety of virtues, it also requires an additional assumption on the covariance of the process. 

\begin{proposition}[Simultaneous Coverage]
    Assume that the estimated covariance function $\hat\Sigma$ is one-time differentiable, and let the function $\hat{\sigma}^2$ denote its main diagonal.
    The $(1-\alpha)$ confidence band
    \begin{equation}
        C_\alpha = [\drfos^l, \drfos^u] = \drfos \pm 
        \frac{1}{\sqrt{n}}
        \hat{u}_{\alpha/2}^*\hat{\sigma}
        %n^{-1/2}
    \end{equation}
    provides simultaneous coverage.
\end{proposition}

Another assumption-lean approach to construct simultaneous confidence bands is based on the parametric bootstrap \citep{pini2017interval}. This approach consists in repeatedly sampling from the process, and estimating quantiles accordingly, without any additional assumption on the covariance function. In symbols, the interval takes the form $C_\alpha = [\hat{\beta}_{\alpha/2}, \hat{\beta}_{1-\alpha/2}]$, where $\hat{\beta}_{\alpha/2}$ and $\hat{\beta}_{1-\alpha/2}$ indicate the estimated $\alpha/2$ and $1-\alpha/2$ quantiles, respectively. Given its ease of implementation, and the fact that it requires no additional assumptions, we adopt the parametric bootstrap approach in the following Sections. 

% \begin{remark}
%     \lorenzo{The choice of the sup-norm means that \DRFOS{}, by design, focuses on global control and the construction of uniform confidence bands, positioning it as a global smoothing method. This contrasts with local methods -- e.g., those using kernel or spline estimators trained point-wise, such as those discussed in \citet{singh2024kernel} -- that may achieve higher pointwise accuracy in $\mathbbmss{L}^2$. However, the double robustness property of \DRFOS{} is crucial for mitigating bias, offering protection against global misspecification, which can be particularly damaging to global estimators.}
% \end{remark}

\begin{remark}
\lorenzo{Our analysis relies on the Banach space $C(\mathcal{T})$ equipped with the sup-norm. This choice is helpful to construct valid simultaneous confidence bands, as it ensures control over the maximum deviation of the estimated treatment effect across the entire domain $\mathcal{T}$. We acknowledge that this entails a stronger topological constraint compared to Hilbert space approaches (e.g., using the $\mathbbmss{L}^2$ norm), where convergence rates require fewer assumptions \citep{van1996weak}. As highlighted by \citet{singh2024kernel}, estimators prioritizing local smoothing can achieve higher pointwise accuracy and may adapt better to local irregularities than global sup-norm approaches. However, such methods generally do not provide the uniform coverage guarantees required for simultaneous inference. By assuming expected Hölder continuity (Assumption \ref{ass:estimate}\textbf{d}), \DRFOS{} secures a parametric $\sqrt{n}$-convergence rate in the supremum norm, effectively trading stronger regularity conditions for the ability to perform valid global inference.}
\end{remark}

\section{Simulation study}
\label{sec:simulate}

We employ synthetic data to investigate and compare performance of \DRFOS{} in both well-specified and misspecified scenarios. We first generate the true regression functions according to the model
\begin{equation}
    \mu^{(a)} = a\beta + \rho\,,\quad a\in\{0,1\}\,,
\end{equation}
where the FATE $\beta$ and the baseline function $\rho$ are drawn from a $0$ mean Gaussian process with a Matern covariance function \citep{cressie1999classes} of the form
\begin{equation}
\label{eq:matern}
        C(s,t) = \frac{\eta^2}{\Gamma(\nu)2^{\nu-1}}\left(\frac{\sqrt{2\nu}}{l} \lvert s-t \rvert\right)^\nu \chi_\nu\left(\frac{\sqrt{2\nu}}{l} \lvert s-t \rvert \right)\,.
\end{equation}
Here $\chi_\nu$ is a modified Bessel function, and we set the parameters to $l=0.25$, $\nu=3.5$, and $\eta^2=1$. 
We then generate the potential outcomes as
\begin{equation}
    \mathcal{Y}^{(a)}_i = \mu^{(a)} + \varepsilon_i\,,
\end{equation}
where the error $\varepsilon_i$ is again drawn from a $0$ mean Gaussian process with a Matern covariance function; here we set $l=0.25$, $\nu=2.5$, and $\eta = \VV{A_i\beta +\rho} / 10$, where with a slight abuse of notation $\VV{A_i\beta + \rho}$ indicates the global (pooled over time) variance of the signal $A_i\beta + \rho$. This roughly ensures that the pooled signal-to-noise ratio is approximately 1. The observed outcome is thus
\begin{equation}
    \mathcal{Y}_i = A_i \mathcal{Y}_i^{(1)} + (1-A_i) \mathcal{Y}_i^{(0)}\,,
\end{equation}
where $A_i\sim\text{Ber}\,(\Tilde{p})$, $\Tilde{p} = 0.5$.

We control the degree of misspecification in the estimated propensity score and in the estimated regression functions through the parameters $\alpha_\pi$ and $\alpha_\mu$, respectively. In particular we define the estimated propensity score as
\begin{equation}
    \hat{\pi}_i = \alpha_\pi U_i + (1-\alpha_\pi) \Tilde{p}\,,
\end{equation}
where $U_i$ is sampled from a balanced mixture of two normal distributions. The first normal has mean $0.2$ and standard deviation $0.1$; the second normal has mean $0.8$ and standard deviation $0.1$. The mixture is truncated in the interval $[0.02, 0.98]$.  
Similarly, we define the estimated regression functions as 
\begin{equation}
    \hat{\mu}^{(a)}_i = \alpha_\mu \mathcal{U}_i + (1-\alpha_\mu) \mu^{(a)}\,,
\end{equation}
where $\mathcal{U}_i$ is drawn from a $0$ mean Gaussian process with Matern covariance function with parameters $l=0.25$, $\nu=2.5$, $\eta^2=1$. This simulation structure allows us to study and compare the performance of the various FATE estimators under different misspecification scenarios, without concerning ourselves with potential biases introduced by the methods used to produce $\hat{\mu}^{(a)}$ and $\hat{\pi}$. In fact, $\alpha_\mu$ and $\alpha_\pi$ control the convex combinations between true and random quantities. The lower the values of $\alpha_\mu$ and $\alpha_\pi$, the stronger the correspondence between true quantities and estimates; the higher the values of $\alpha_\mu$ and $\alpha_\pi$, the more the estimates are pulled towards random directions, progressively injecting bias. In a way, $\alpha_\pi$ and $\alpha_\mu$ can be interpreted as proxies of the rate of convergence of $\hat{\pi}$ and $\hat{\mu}^{(a)}$ to their corresponding population quantities -- see \citet{das2024doubly} for a more explicit characterization.
%\francesca{[interesting. we are simulating misspecifications that do {\em not} induce bias; how realistic is this? misspecifications are usually thought of as biasing, aside from the esimators employed]} \lorenzo{[Why do you think these misspecifications do not introduce bias? When $\alpha_\mu$ and $\alpha_\pi$ are big, nuisance estimators $\hat{\pi}$ and $\hat{\mu}$ are completely biased towards random quantities. What we are avoiding with this strategy is the systematic bias of, say, (i) generate data that a (ii) neural network model always handles in the same way for mysteryous reasons.]} 
For each simulation scenario, we set $n=5000$, $\alpha_\pi,\alpha_\mu \in \{0, 0.25, 0.5, 0.75, 1\}$ and run experiments with $50$ different seeds. 

We perform extensive simulations to compare \DRFOS{} with two commonly used FATE estimators: the \textit{outcome regression} estimator defined as 
\begin{equation}
\label{eq:or}
    \hat{\beta}_{\texttt{OR}} = \frac{1}{n} \sum_{i=1}^n \left[\hat{\mu}^{(1)}_i - \hat{\mu}^{(0)}_i \right]\,,
\end{equation}
and the \textit{inverse probability weighting} estimator defined as
\begin{equation}
\label{eq:ipw}
    \hat{\beta}_{\texttt{IPW}} = \frac{1}{n} \sum_{i=1}^n \left[ \frac{A_i\mathcal{Y}_i}{\hat{\pi}^{(1)}_i} - \frac{(1-A_i) \mathcal{Y}_i}{1 - \hat{\pi}^{(1)}_i}\right]\,.
\end{equation}
The latter estimator is well known in the scalar setting, but to the best of our knowledge it has never been analyzed in the functional data literature. In Appendix Section~\ref{suppsec:fipw} we therefore study its properties and asymptotic behavior.

We assess performance in terms of estimation accuracy and inferential coverage. To measure accuracy, we compute the mean squared error under the $\mathbbmss{L}^2$ distance between the true $\beta$ and its estimate $\hat{\beta}$, i.e.~$\text{MSE}(\hat{\beta}) = \int (\beta(t) - \hat{\beta}(t))^2\,dt$. To measure actual coverage, we evaluate the percentage $\Delta$ of the domain of $\beta$ that falls into the parametric bootstrap simultaneous coverage bands, i.e.~$\Delta = \int \one \{\beta(t)\in [\hat{\beta}^l(t), \hat{\beta}^u(t)] \} \,dt$.

\begin{figure}[t]
    \centering
    \includegraphics[width=\linewidth]{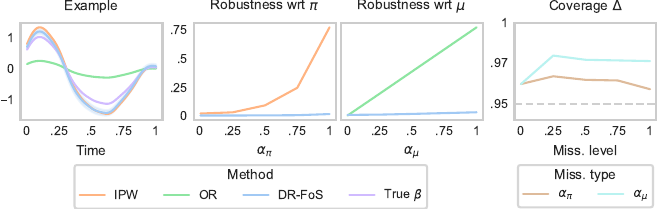}
    \caption{Simulation results based on the data generating process described in Section~\ref{sec:simulate}. The first panel (leftmost) shows an example of true FATE $\beta$ (purple) as generated by our simulation scheme, and the estimates provided by \texttt{IPW} (orange), \texttt{OR} (green) and \DRFOS{} (blue) -- together with simultaneous confidence bands around the latter. In this example we set $\alpha_\mu=0.25$ and $\alpha_\pi=0.75$. The second panel shows the average (across simulation replicates) estimation error as captured by the $\text{MSE}$ when the estimated propensity score is increasingly corrupted by random noise, and the regression functions are well-specified ($\alpha_\mu=0$, $\alpha_\pi\neq0$). While \DRFOS{} (blue) maintains excellent performance, \texttt{IPW} (orange) performs progressively worse, as expected. The third panel displays the opposite situation, where the estimated regression functions are increasingly corrupted by random noise and the propensity score is well-specified ($\alpha_\mu\neq0$, $\alpha_\pi=0$). Again, the performance of \DRFOS{} (blue) remains excellent, while the performance of \texttt{OR} (green) progressively deteriorates. Finally, the last panel (rightmost) shows the average coverage levels achieved under misspecifications of the regression functions (light blue, $\alpha_\pi=0$) and of the propensity score (brown, $\alpha_\mu=0$). The grey dashed line represents the nominal $95\%$ coverage level.}
    \label{fig:sim}
\end{figure}

Figure~\ref{fig:sim} shows an example of true FATE alongside its estimates using \texttt{OR}, \texttt{IPW}, and \DRFOS{}, for one specific scenario (first panel; the underlying simulation parameters are reported in the caption). Figure~\ref{fig:sim} also displays the performance of the various estimators based on the metrics described above. In a well-specified scenario, all estimators are comparable. However, when random noise corrupts either the propensity score (second panel) or the regression function (third panel), the performance of \texttt{IPW} and \texttt{OR} markedly deteriorates. In contrast, \DRFOS{} consistently achieves high accuracy in both cases, provided at least one of the two models -- the propensity score or the regression function -- is well-specified.
The fourth panel displays coverage performance. The simultaneous confidence bands effectively control the type I error, ensuring the nominal 95\% coverage level. However, under both misspecification scenarios, \DRFOS{} appears to exhibit some overcoverage. In Appendix Section~\ref{suppsec:simulations} we report results for additional simulation scenarios where both models are misspecified at the same time \lorenzo{and where we introduce discontinuities to evaluate the behavior of \DRFOS{} simultaneous confidence bands when the continuity assumption is not satisfied.} There, we also report further, more standard simulation scenarios where we explicitly model the data generating process and the nuisance functions, such as the linear case in which $\mu^{(a)}(X_i) = \sum_{j=1}^p X_{ij}\theta_j^{(a)}$ for some scalar variables $X_{ij}$ and functional coefficients $\theta_j^{(a)}$; notably, even in these cases, \DRFOS{} matches or outperforms the other estimators.

\section{SHARE application}
\label{sec:apply}
To demonstrate the use of \DRFOS{}, we apply it to the SHARE ({\em Survey of Health, Aging and Retirement in Europe}) dataset \citep{mannheim2005survey}. SHARE is a research infrastructure that aims to investigate the effects of health, social, economic, and environmental policies on the life course of European citizens \citep{bergmann2017survey, borsch2013data, borsch2020survey}. SHARE is a longitudinal study, where the same subjects are followed over multiple years. Specifically, eight surveys -- or ``waves'' -- were conducted from 2004 to 2020.

Specifically, we employ the proposed estimator to investigate the causal effect of chronic conditions on longitudinal indicators of quality of life. We first preprocess data following the steps described in \citet{boschi2024new}. We focus on the 1518 subjects who participated in at least seven out of the eight waves, ensuring a sufficient number of measurements for reliable curve estimation.
We investigate a subset of the variables from the EasySHARE dataset \citep{gruber2014generating}, a preprocessed version of the SHARE data. For each subject, we consider: two functional indicators of quality of life -- a mobility index and the Quality of Life Scale (CASP) \citep{hyde2003measure}, measured over 192 months; two chronic diseases -- high cholesterol and hypertension; and a variety of socio-demographic and healthcare factors (see Table~\ref{tab:covariates}). 
While some of these variables are characterized by values that change over time (e.g., CASP and mobility index) and are suitable for a functional representation, others are scalar (e.g., education years) or categorical (e.g., gender) and do not evolve over different waves. 
We smooth time-varying variables using cubic \emph{B-splines} with knots at each survey date and roughness penalty on the curves second derivative \citep{ramsay2005}. For each curve, the smoothing parameter is selected by minimizing the average generalized cross-validation error \citep{craven1978smoothing}. 
Although the survey dates and the amount of measurements may vary across subjects, the functional representation provides a natural imputation for missing values and facilitates the comparison of different statistical units across the entire temporal domain.

\begin{table}
    \centering
    \caption{Variables employed in the SHARE application study.}
    \label{tab:covariates}
    \begin{tabular}{lll}
        \toprule
        \textbf{Variable name} & \textbf{Type} & \textbf{Usage} \\
        \midrule
        Hypertension condition & Binary & Treatment \\
        High cholesterol & Binary & Treatment \\
        Age & Scalar & Covariate \\
        Years of education & Scalar & Covariate \\
        Score of first numeracy test & Scalar & Covariate \\
        Gender & Binary & Covariate \\
        Smoke & Binary & Covariate \\
        Vaccinations during childwood & Binary & Covariate \\
        Mobility index & Functional & Outcome \\
        CASP & Functional & Outcome \\
        \bottomrule
    \end{tabular}
\end{table}

Following the logic of the counterfactual model, we define the causal effect of a chronic disease on a quality of life indicator as the expected difference between the functional outcome if a subject had, or did not have, the chronic disease. Therefore, we first form treatment groups identifying subjects who present chronic conditions at the beginning of the study. Similarly, we form control groups identifying subjects who never develop chronic conditions throughout the study period. \lorenzo{We acknowledge that this retrospective definition of the control group introduces a form of selection or immortal-time bias since it utilizes information post-baseline (the absence of the condition throughout the study). Therefore, for this illustrative application, the estimand should be interpreted as the average difference in functional trajectories conditional on long-term disease-free status versus chronic disease status at baseline, rather than a causal effect of incident treatment. We acknowledge that more complex causal frameworks, like g-computation or Marginal Structural Models (MSMs), can deal with time-varying treatments or censoring by handling time-dependent confounding and estimating dynamic treatment effects \citep{robins1986new, tsiatis2006semiparametric}. However, these approaches cannot be readily applied to functional data.} For high cholesterol, we have 313 subjects in the treatment group and 747 in the control group. Similarly, for hypertension, we have 419 treated subjects and 577 controls.

For each combination of chronic disease and functional outcome, we then fit \DRFOS{} using a function-on-scalar least squares specification for the regression function $\hat{\mu}^{(a)}$, and logistic regression for the propensity score $\hat{\pi}^{(a)}$ (additional implementations with more sophisticated nuisance models can be found in Appendix Section~\ref{suppsec:more_real}). \lorenzo{No time-varying or post-baseline variables beyond the cohort definition are used in the propensity score or outcome models.} We employ cross-fitting with 5 balanced folds. Results are shown in Figure~\ref{fig:app}. The two chronic conditions both have a positive and statistically significant effect on the mobility index across the entire time period. Thus, chronic conditions adversely affect mobility (a higher mobility index corresponds to reduced mobility). Similarly, the effects of both chronic conditions on CASP are negative and statistically significant throughout the time period, highlighting a detrimental impact on quality of life as measured by this metric. 
The impacts of chronic conditions on functional outcomes also appear to increase in magnitude over time. Thus, as individuals age, the adverse effects of chronic conditions on their quality of life become progressively more pronounced. Taken together, these findings demonstrate the utility of \DRFOS{} in real-world applications, and its ability to enhance our understanding of complex phenomena.

\begin{figure}
    \centering
    \includegraphics[width=\linewidth]{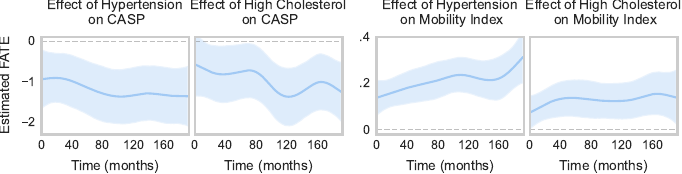}
    \caption{SHARE application results. Each panel displays a different estimated causal effect. Blue continuous lines correspond to \DRFOS{} estimates; blue bands are 95\% asymptotic simultaneous confidence bands obtained by repeatedly sampling from the Gaussian process; grey dotted horizontal lines correspond to 0.}
    \label{fig:app}
\end{figure}

\section{Conclusion}
\label{sec:conclude}

In this paper, we introduced \DRFOS{}, a novel estimator for the Functional Average Treatment Effect (FATE) in observational settings with functional outcomes. By leveraging double robustness, \DRFOS{} offers consistent estimation even when the outcome model, or alternatively the treatment assignment model, is misspecified. This desirable feature ensures that our methodology is robust to potential inaccuracies in the modeling assumptions, a critical advantage in practical applications where accurate, or even just reasonable model specifications are rarely guaranteed.

Our contributions are both theoretical and relevant for practitioners. We rigorously establish the theoretical properties of \DRFOS{}, proving its convergence to a Gaussian process under weak assumptions. This allows us to pursue pointwise inference, as well as to construct simultaneous confidence bands, providing a comprehensive inferential framework for functional treatment effects. The methodology is versatile, leveraging modern techniques in functional data analysis and causal inference while maintaining computational efficiency. 

Through extensive simulations, we show that \DRFOS{} outperforms commonly used alternatives, particularly under scenarios of model misspecification. These results highlight the practical utility of our approach in functional data settings, offering researchers a reliable tool for causal inference in complex, high-dimensional domains. Additionally, the application to the SHARE dataset underscores the flexibility and relevance of \DRFOS{} in addressing scientific questions that involve functional outcomes.

We envision several directions for future work. First, extending the methodology to more complex causal structures, such as scalar-on-function or function-on-function causal relationships, where the treatment is a continuous variable, represents a promising direction for further investigation \citep{tan2025causal}. Additionally, the assumptions underpinning the convergence to a Gaussian process could be further relaxed, particularly in the context of non i.i.d.~data or outcomes observed on non-standard domains. These extensions would broaden the applicability of \DRFOS{} and further strengthen its theoretical foundations.

In summary, \DRFOS{} provides a robust, efficient, and theoretically grounded approach for estimating causal effects in observational studies where outcomes are functional data. By bridging the gap between classic causal inference and functional data analysis, it sets the stage for further advancements in the analysis of complex and structured outcomes in modern scientific applications.

% Acknowledgements and Disclosure of Funding should go at the end, before appendices and references

\acks{L.T.~wishes to thank Jing Lei for helpful discussions about the theoretical properties of the proposed estimator, and Antonio C.~Herling Ribeiro Jr.~for invaluable feedback on the simulation strategy. The work of F.C.~was partially supported by the Huck Institutes of the Life Sciences at Penn State, the L'EMbeDS Department of Excellence of the Sant'Anna School of Advanced Studies, and the SMaRT COnSTRUCT project (CUP J53C24001460006, as part of FAIR, PE0000013, CUP B53C22003630006, Italian National Recovery and Resilience Plan funded by NextGenerationEU).
The work of E.H.K.~was supported by the NSF CAREER Award 2047444. The work of M.R.~was partially supported by the NSF SES Award 1853209. 
}

\vskip 0.2in
\bibliography{sample}

@article{ecker2024causal,
  title={Causal inference with a functional outcome},
  author={Ecker, Kreske and de Luna, Xavier and Schelin, Lina},
  journal={Journal of the Royal Statistical Society Series C: Applied Statistics},
  volume={73},
  number={1},
  pages={221--240},
  year={2024},
  publisher={Oxford University Press US}
}

@article{liebl2023fast,
  title={Fast and fair simultaneous confidence bands for functional parameters},
  author={Liebl, Dominik and Reimherr, Matthew},
  journal={Journal of the Royal Statistical Society Series B: Statistical Methodology},
  volume={85},
  number={3},
  pages={842--868},
  year={2023},
  publisher={Oxford University Press US}
}

@article{lin2023causal,
  title={Causal inference on distribution functions},
  author={Lin, Zhenhua and Kong, Dehan and Wang, Linbo},
  journal={Journal of the Royal Statistical Society Series B: Statistical Methodology},
  volume={85},
  number={2},
  pages={378--398},
  year={2023},
  publisher={Oxford University Press US}
}

@article{pini2017interval,
  title={Interval-wise testing for functional data},
  author={Pini, Alessia and Vantini, Simone},
  journal={Journal of Nonparametric Statistics},
  volume={29},
  number={2},
  pages={407--424},
  year={2017},
  publisher={Taylor \& Francis}
}

@book{billingsley2013convergence,
  title={Convergence of probability measures},
  author={Billingsley, Patrick},
  year={2013},
  publisher={John Wiley \& Sons}
}

@article{dette2020functional,
  title={Functional data analysis in the Banach space of continuous functions},
  author={Dette, Holger and Kokot, Kevin and Aue, Alexander},
  journal={The Annals of Statistics},
  volume={48},
  number={2},
  pages={1168--1192},
  year={2020},
  publisher={JSTOR}
}

@book{bosq2000linear,
  title={Linear processes in function spaces: theory and applications},
  author={Bosq, Denis},
  volume={149},
  year={2000},
  publisher={Springer Science \& Business Media}
}

@book{ferraty2006nonparametric,
  title={Nonparametric functional data analysis},
  author={Ferraty, Fr{\'e}d{\'e}ric},
  year={2006},
  publisher={Springer}
}

@book{kokoszka2017introduction,
  title={Introduction to functional data analysis},
  author={Kokoszka, Piotr and Reimherr, Matthew},
  year={2017},
  publisher={CRC Press}
}

@book{ramsay2005,
	title={Functional data analysis},
	author={Ramsay, James O. and   Silverman, B. W.},
	edition={2},
	publisher={Springer},
	year={2005}
}

@article{boschi2024fasten,
  title={Fasten: an efficient adaptive method for feature selection and estimation in high-dimensional functional regressions},
  author={Boschi, Tobia and Testa, Lorenzo and Chiaromonte, Francesca and Reimherr, Matthew},
  journal={Journal of Computational and Graphical Statistics},
  number={just-accepted},
  pages={1--24},
  year={2024},
  publisher={Taylor \& Francis}
}

@article{cressie1999classes,
  title={Classes of nonseparable, spatio-temporal stationary covariance functions},
  author={Cressie, Noel and Huang, Hsin-Cheng},
  journal={Journal of the American Statistical Association},
  volume={94},
  number={448},
  pages={1330--1339},
  year={1999},
  publisher={Taylor \& Francis Group}
}

@article{kennedy2022semiparametric,
  title={Semiparametric doubly robust targeted double machine learning: a review},
  author={Kennedy, Edward H},
  journal={Handbook of Statistical Methods for Precision Medicine},
  pages={207--236},
  year={2024},
  publisher={Chapman and Hall/CRC}
}

@article{jeong2024identifying,
  title={Identifying sparse treatment effects},
  author={Jeong, Yujin and Fox, Emily and Johari, Ramesh},
  journal={arXiv preprint arXiv:2404.14644},
  year={2024}
}

@article{kennedy2019estimating,
  title={Estimating scaled treatment effects with multiple outcomes},
  author={Kennedy, Edward H and Kangovi, Shreya and Mitra, Nandita},
  journal={Statistical methods in medical research},
  volume={28},
  number={4},
  pages={1094--1104},
  year={2019},
  publisher={SAGE Publications Sage UK: London, England}
}

@article{kennedy2023semiparametric,
  title={Semiparametric counterfactual density estimation},
  author={Kennedy, Edward H and Balakrishnan, Sivaraman and Wasserman, LA},
  journal={Biometrika},
  volume={110},
  number={4},
  pages={875--896},
  year={2023},
  publisher={Oxford University Press}
}

@article{luedtke2024one,
  title={One-step estimation of differentiable Hilbert-valued parameters},
  author={Luedtke, Alex and Chung, Incheoul},
  journal={The Annals of Statistics},
  volume={52},
  number={4},
  pages={1534--1563},
  year={2024},
  publisher={Institute of Mathematical Statistics}
}

@article{kennedy2020optimal,
  title={Optimal doubly robust estimation of heterogeneous causal effects},
  author={Kennedy, Edward H and others},
  journal={arXiv preprint arXiv:2004.14497},
  volume={5},
  year={2020}
}

@article{boschi2021functional,
  title={Functional data analysis characterizes the shapes of the first COVID-19 epidemic wave in Italy},
  author={Boschi, Tobia and Di Iorio, Jacopo and Testa, Lorenzo and Cremona, Marzia A and Chiaromonte, Francesca},
  journal={Scientific reports},
  volume={11},
  number={1},
  pages={17054},
  year={2021},
  publisher={Nature Publishing Group UK London}
}

@article{boschi2023contrasting,
  title={Contrasting pre-vaccine COVID-19 waves in Italy through Functional Data Analysis},
  author={Boschi, Tobia and Di Iorio, Jacopo and Testa, Lorenzo and Cremona, Marzia A and Chiaromonte, Francesca},
  journal={arXiv preprint arXiv:2307.09820},
  year={2023}
}

@inproceedings{boschi2024new,
  title={A New Computationally Efficient Algorithm to solve Feature Selection for Functional Data Classification in High-dimensional Spaces},
  author={Boschi, Tobia and Bonin, Francesca and Ordonez-Hurtado, Rodrigo and Pascale, Alessandra and Epperlein, Jonathan P},
  booktitle={Forty-first International Conference on Machine Learning},
  year={2024}
}

@article{qi2018function,
  title={Function-on-function regression with thousands of predictive curves},
  author={Qi, Xin and Luo, Ruiyan},
  journal={Journal of Multivariate Analysis},
  volume={163},
  pages={51--66},
  year={2018},
  publisher={Elsevier}
}

@article{prosperi2020causal,
  title={Causal inference and counterfactual prediction in machine learning for actionable healthcare},
  author={Prosperi, Mattia and Guo, Yi and Sperrin, Matt and Koopman, James S and Min, Jae S and He, Xing and Rich, Shannan and Wang, Mo and Buchan, Iain E and Bian, Jiang},
  journal={Nature Machine Intelligence},
  volume={2},
  number={7},
  pages={369--375},
  year={2020},
  publisher={Nature Publishing Group UK London}
}

@article{varian2016causal,
  title={Causal inference in economics and marketing},
  author={Varian, Hal R},
  journal={Proceedings of the National Academy of Sciences},
  volume={113},
  number={27},
  pages={7310--7315},
  year={2016},
  publisher={National Acad Sciences}
}

@article{imbens2024causal,
  title={Causal inference in the social sciences},
  author={Imbens, Guido W},
  journal={Annual Review of Statistics and Its Application},
  volume={11},
  year={2024},
  publisher={Annual Reviews}
}

@article{van2006targeted,
  title={Targeted maximum likelihood learning},
  author={Van Der Laan, Mark J and Rubin, Daniel},
  journal={The international journal of biostatistics},
  volume={2},
  number={1},
  year={2006},
  publisher={De Gruyter}
}

@book{tsiatis2006semiparametric,
  title={Semiparametric theory and missing data},
  author={Tsiatis, Anastasios A},
  volume={4},
  year={2006},
  publisher={Springer}
}

@article{cremona2018iwtomics,
  title={IWTomics: testing high-resolution sequence-based ‘Omics’ data at multiple locations and scales},
  author={Cremona, Marzia A and Pini, Alessia and Cumbo, Fabio and Makova, Kateryna D and Chiaromonte, Francesca and Vantini, Simone},
  journal={Bioinformatics},
  volume={34},
  number={13},
  pages={2289--2291},
  year={2018},
  publisher={Oxford University Press}
}

@inproceedings{boschi2024fungcn,
  title     = {Functional Graph Convolutional Networks: A Unified Multi-task and Multi-modal Learning Framework to Facilitate Health and Social-Care Insights},
  author    = {Boschi, Tobia and Bonin, Francesca and Ordonez-Hurtado, Rodrigo and Rousseau, Cécile and Pascale, Alessandra and Dinsmore, John},
  booktitle = {Proceedings of the Thirty-Third International Joint Conference on
               Artificial Intelligence, {IJCAI-24}},
  publisher = {International Joint Conferences on Artificial Intelligence Organization},
  pages     = {7188--7196},
  year      = {2024},
}

@book{van2000asymptotic,
  title={Asymptotic statistics},
  author={Van der Vaart, Aad W},
  volume={3},
  year={2000},
  publisher={Cambridge university press}
}

@book{mannheim2005survey,
  title={The Survey of Health, Aging, and Retirement in Europe: Methodology},
  author={
Alcser, Kirsten H. and
Benson, Grant and 
Barsch-Supan, Axel and
Brugiavini, Agar and
Christelis, Dimitrios and 
Croda, Enrica and 
Das, Marcel and
de Luca, Guiseppe and 
Harkness, Janet and
Hesselius, Patrik and
Jappelli, Tullio and 
Jarges, Hendrik and
Kalwij, Adriaan and 
Kemperman, Marie-Louise and
Klevmarken, Anders and
Lipps, Oliver and
Paccagnella, Omar and
Padula, Mario and
Perrachi, Franco and
Rainato, Roberta and
van Soest, Arthur and
Swensson, Bengt and
Vis, Corrie and
Weber, Guglielmo and
Weerman, Bas
  },
  year={2005},
  publisher={Mannheim Research Institute for the Economics of Aging (MEA)}
}

@article{borsch2013data,
  title={Data resource profile: the Survey of Health, Ageing and Retirement in Europe (SHARE)},
  author={B{\"o}rsch-Supan, Axel and Brandt, Martina and Hunkler, Christian and Kneip, Thorsten and Korbmacher, Julie and Malter, Frederic and Schaan, Barbara and Stuck, Stephanie and Zuber, Sabrina},
  journal={International journal of epidemiology},
  volume={42},
  number={4},
  pages={992--1001},
  year={2013},
  publisher={Oxford University Press}
}

@article{bergmann2017survey,
  title={Survey participation in the survey of health, ageing and retirement in Europe (SHARE), Wave 1-6},
  author={Bergmann, Michael and Kneip, Thorsten and De Luca, Giuseppe and Scherpenzeel, Annette},
  journal={Munich: Munich Center for the Economics of Aging},
  year={2017}
}

@article{borsch2020survey,
  title={Survey of health, ageing and retirement in Europe (SHARE) wave 5},
  author={B{\"o}rsch-Supan, Axel},
  journal={Release version},
  volume={7},
  number={0},
  year={2020}
}

@techreport{gruber2014generating,
  title={Generating easySHARE: guidelines, structure, content and programming},
  author={Gruber, Stefan and Hunkler, Christian and Stuck, Stephanie},
  year={2014},
  institution={SHARE Working Paper Series 17-2014. Munich}
}

@article{craven1978smoothing,
  title={Smoothing noisy data with spline functions},
  author={Craven, Peter and Wahba, Grace},
  journal={Numerische mathematik},
  volume={31},
  number={4},
  pages={377--403},
  year={1978},
  publisher={Springer}
}

@article{liu2024double,
  title={Double robust estimation of functional outcomes with data missing at random},
  author={Liu, Xijia and Ecker, Kreske and Schelin, Lina and de Luna, Xavier},
  journal={arXiv preprint arXiv:2411.17224},
  year={2024}
}

@article{das2024doubly,
  title={Doubly robust capture-recapture methods for estimating population size},
  author={Das, Manjari and Kennedy, Edward H and Jewell, Nicholas P},
  journal={Journal of the American Statistical Association},
  volume={119},
  number={546},
  pages={1309--1321},
  year={2024},
  publisher={Taylor \& Francis}
}

@article{hyde2003measure,
  title={A measure of quality of life in early old age: the theory, development and properties of a needs satisfaction model (CASP-19)},
  author={Hyde, Martin and Wiggins, Richard D and Higgs, Paul and Blane, David B},
  journal={Aging \& mental health},
  volume={7},
  number={3},
  pages={186--194},
  year={2003},
  publisher={Taylor \& Francis}
}

@inproceedings{pavlyshenko2018using,
  title={Using stacking approaches for machine learning models},
  author={Pavlyshenko, Bohdan},
  booktitle={2018 IEEE second international conference on data stream mining \& processing (DSMP)},
  pages={255--258},
  year={2018},
  organization={IEEE}
}

@article{van2007super,
  title={Super learner},
  author={Van der Laan, Mark J and Polley, Eric C and Hubbard, Alan E},
  journal={Statistical applications in genetics and molecular biology},
  volume={6},
  number={1},
  year={2007},
  publisher={De Gruyter}
}

@misc{wager2024causal,
  title={Causal inference: A statistical learning approach},
  author={Wager, Stefan},
  year={2024},
  publisher={preparation}
}

@inproceedings{ieva2025enhancing,
  title={Enhancing Causal Inference in Functional Data: a Method for Estimating Time-Varying Causal Treatment Effects},
  author={Ieva, Francesca and Fontana, Nicole and Pivato, Carlo Andrea and Di Angelantonio, Emanuele and Secchi, Piercesare},
  booktitle={International Workshop on Functional and Operatorial Statistics},
  pages={285--293},
  year={2025},
  organization={Springer}
}

@article{martinez2023efficient,
  title={An efficient doubly-robust test for the kernel treatment effect},
  author={Martinez Taboada, Diego and Ramdas, Aaditya and Kennedy, Edward},
  journal={Advances in Neural Information Processing Systems},
  volume={36},
  pages={59924--59952},
  year={2023}
}

@article{lundborg2022conditional,
  title={Conditional independence testing in Hilbert spaces with applications to functional data analysis},
  author={Lundborg, Anton Rask and Shah, Rajen D and Peters, Jonas},
  journal={Journal of the Royal Statistical Society Series B: Statistical Methodology},
  volume={84},
  number={5},
  pages={1821--1850},
  year={2022},
  publisher={Oxford University Press}
}

@article{singh2024kernel,
  title={Kernel methods for causal functions: dose, heterogeneous and incremental response curves},
  author={Singh, Rahul and Xu, Liyuan and Gretton, Arthur},
  journal={Biometrika},
  volume={111},
  number={2},
  pages={497--516},
  year={2024},
  publisher={Oxford University Press}
}

@article{bruns2025augmented,
  title={Augmented balancing weights as linear regression},
  author={Bruns-Smith, David and Dukes, Oliver and Feller, Avi and Ogburn, Elizabeth L},
  journal={Journal of the Royal Statistical Society Series B: Statistical Methodology},
  pages={qkaf019},
  year={2025},
  publisher={Oxford University Press UK}
}

@article{raykov2025kernel,
  title={Kernel-based estimators for functional causal effects},
  author={Raykov, Yordan P and Luo, Hengrui and Strait, Justin D and KhudaBukhsh, Wasiur R},
  journal={arXiv preprint arXiv:2503.05024},
  year={2025}
}

@article{thind2020funcnn,
  title={Funcnn: An r package to fit deep neural networks using generalized input spaces},
  author={Thind, Barinder and Wu, Sidi and Groenewald, Richard and Cao, Jiguo},
  journal={arXiv preprint arXiv:2009.09111},
  year={2020}
}

@article{robins1986new,
  title={A new approach to causal inference in mortality studies with a sustained exposure period—application to control of the healthy worker survivor effect},
  author={Robins, James},
  journal={Mathematical modelling},
  volume={7},
  number={9-12},
  pages={1393--1512},
  year={1986},
  publisher={Elsevier}
}

@incollection{van1996weak,
  title={Weak convergence},
  author={Van Der Vaart, Aad W and Wellner, Jon A},
  booktitle={Weak convergence and empirical processes: with applications to statistics},
  pages={16--28},
  year={1996},
  publisher={Springer}
}

@article{tan2025causal,
  title={Causal effect of functional treatment},
  author={Tan, Ruoxu and Huang, Wei and Zhang, Zheng and Yin, Guosheng},
  journal={Journal of Machine Learning Research},
  volume={26},
  number={91},
  pages={1--39},
  year={2025}
}


\newpage

\appendix

\section{Proofs of main results}
\label{suppsec:proofs}

Before presenting the proofs of the main results, we introduce additional notation to facilitate the exposition. First, for any $k\in\NN$ and $t_1,\dots,t_k\in\mathcal{T}$, we denote the $k$-dimensional projection vectors of $\beta$ and $\drfos$ at $(t_1,\dots,t_k)$ as $B_{t_1, \dots, t_k} = \left(\beta(t_1),\dots,\beta(t_k)\right)^T$ and $\hat{B}_{t_1, \dots, t_k} = \left(\drfos(t_1),\dots,\drfos(t_k)\right)^T$, respectively. Similarly, we denote the $k$-dimensional projection vectors of $\influence{\data}$ and $\hatinfluence{\data}$ at $(t_1,\dots,t_k)$ as $\projection{t_1,\dots,t_k}{\data} = (\influence{\data;t_1}, \dots, \influence{\data;t_k})^T$ and $\hatprojection{t_1,\dots,t_k}{\data} = (\hatinfluence{\data;t_1}, \dots, \hatinfluence{\data;t_k})^T$.

\subsection{Lemma \ref{th:equivalence}}
We start by showing that $\EE{\mathcal{Y}^{(1)}} = \EE{\gamma^{(1)}(\data)}$. This is done exploiting the identifiability Assumptions~\ref{ass:identify}. Indeed
\begin{equation}
    \begin{split}
        \EE{\gamma^{(1)}(\data)} &=\EE{\mu^{(1)}(X) + \frac{A(\mathcal{Y}-\mu^{(1)}(X))}{\pi^{(1)}(X)}} \\
        &= \EE{\mu^{(1)}(X)} + \EE{\frac{A(\mathcal{Y}-\mu^{(1)}(X))}{\pi^{(1)}(X)}} \\
        &= \EE{\EE{\mathcal{Y}\,|\,X,A=1}} + \EE{\EE{\frac{A(\mathcal{Y}-\mu^{(1)}(X))}{\pi^{(1)}(X)}\,\Big|\,X}} \\
        &= \EE{\EE{\mathcal{Y}\,|\,X,A=1}} \\
        &= \EE{\mathcal{Y}^{(1)}}\,.
    \end{split}
\end{equation}
A similar computation leads to $\EE{\mathcal{Y}^{(0)}} = \EE{\gamma^{(0)}(\data)}$. Combining the two results implies the desired equality.

\subsection{Lemma \ref{lem:pan}}
For simplicity, we show the result for a single cross-fitting fold. In particular, denote the distribution where $\hat{\mu}^{(a)}$ and $\hat{\pi}^{(a)}$ are trained as $\hat{\mathbbmss{P}}$, and the distribution where the influence functions are approximated as $\mathbbmss{P}_n$.

First, assume that $\Sigma_{t_1,\dots,t_k} = \EE{ \projection{t_1,\dots,t_k}{\data} \projection{t_1,\dots,t_k}{\data}^T}$ is known. By the Von Mises expansion, we have:
\begin{equation}
   \sqrt{n} \left( \hat{B}_{t_1,\dots,t_k} - B_{t_1,\dots,t_k} \right) = \frac{1}{\sqrt{n}}\sum_{i=1}^n \projection{t_1,\dots,t_k}{\data_i} + \sqrt{n}(\mathbbmss{P}_n - \mathbbmss{P}) \left(\hatprojection{t_1,\dots,t_k}{\data} - \projection{t_1,\dots,t_k}{\data} \right) + \sqrt{n}R_2\,,
\end{equation}
where $\hatprojection{t_1,\dots,t_k}{\data} = (\hat{\varphi}(\data;t_1),\dots,\hat{\varphi}(\data;t_k))^T$, with $\hat{\varphi} = \hat{\gamma}^{(1)} - \hat{\gamma}^{(0)} - \hat{\beta}$. We analyze each component independently. The first term on the right hand side is a sum of mean 0, finite variance random vectors, with the right $n^{-1/2}$ scaling, and by the Central Limit Theorem this converges to a multivariate normal distribution with mean 0 and covariance matrix equal to $\Sigma_{t_1,\dots,t_k}$.

We now provide a bound to the second term, which is an \textit{empirical process}. We need to show that this empirical process, which is a random vector, converges in probability to the 0 vector. By Theorem 2.7 in \citet{van2000asymptotic}, \lorenzo{convergence of this random vector to 0 is equivalent} to component-wise convergence in probability to 0. We thus want to show that each component of the second term is of order $o_\mathbbmss{P}(1)$, as long as the assumption of consistency of the projection of the influence functions holds \citep{kennedy2020optimal}. Conditional on $\hat{\mathbbmss{P}}$, each component of this empirical process has mean 0. In fact, considering an arbitrary component $q\in\{1,\dots,k\}$, we have
\begin{equation}
\begin{split}
    \EE{\mathbbmss{P}_n \left(\richhatprojection{t_1,\dots,t_k}{q}{\data} - \richprojection{t_1,\dots,t_k}{q}{\data} \right)\big|\, \hat{\mathbbmss{P}}} &=  \EE{\richhatprojection{t_1,\dots,t_k}{q}{\data} - \richprojection{t_1,\dots,t_k}{q}{\data}\big|\, \hat{\mathbbmss{P}}} \\
    &= \PP{\richhatprojection{t_1,\dots,t_k}{q}{\data} - \richprojection{t_1,\dots,t_k}{q}{\data}}\,.
\end{split}
\end{equation}
Also, the conditional variance is bounded by 
\begin{equation}
    \begin{split}
        \VV{(\mathbbmss{P}_n - \mathbbmss{P}) \left(\richhatprojection{t_1,\dots,t_k}{q}{\data} - \richprojection{t_1,\dots,t_k}{q}{\data} \right)\big|\, \hat{\mathbbmss{P}}} &= \VV{\mathbbmss{P}_n \left(\richhatprojection{t_1,\dots,t_k}{q}{\data} - \richprojection{t_1,\dots,t_k}{q}{\data} \right)\big|\, \hat{\mathbbmss{P}}} \\
        &= \frac{1}{n} \VV{\richhatprojection{t_1,\dots,t_k}{q}{\data} - \richprojection{t_1,\dots,t_k}{q}{\data}\big|\, \hat{\mathbbmss{P}}} \\
        &\leq \frac{1}{n} \EE{\left(\richhatprojection{t_1,\dots,t_k}{q}{\data} - \richprojection{t_1,\dots,t_k}{q}{\data} \right)^2}\,.
    \end{split}
\end{equation}
Therefore, by iterated expectation and Chebyshev's inequality, we have
\begin{equation}
    \PP{\frac{\sqrt{n} \left| (\mathbbmss{P}_n - \mathbbmss{P}) \left(\richhatprojection{t_1,\dots,t_k}{q}{\data} - \richprojection{t_1,\dots,t_k}{q}{\data} \right) \right|}{\EE{\left(\richhatprojection{t_1,\dots,t_k}{q}{\data} - \richprojection{t_1,\dots,t_k}{q}{\data} \right)^2}} \geq t} \leq \frac{1}{t^2}\,.
\end{equation}
Thus, for any $\varepsilon>0$, there exists a $t=\varepsilon^{-1/2}$ such that the probability above is controlled by $\varepsilon$. This yields the result
\begin{equation}
    (\mathbbmss{P}_n - \mathbbmss{P}) \left(\richhatprojection{t_1,\dots,t_k}{q}{\data} - \richprojection{t_1,\dots,t_k}{q}{\data} \right) = O_\mathbbmss{P}\left(\frac{\EE{\left(\richhatprojection{t_1,\dots,t_k}{q}{\data} - \richprojection{t_1,\dots,t_k}{q}{\data} \right)^2}}{\sqrt{n}}\right)\,.
\end{equation}
In turn, because we assume consistency of the projections of influence functions, this implies the desired result. 

Finally, by assumption, each component in the third term, the \textit{remainder} of the expansion, is of order $o_\mathbbmss{P}(1)$. Summarizing the previous results, we have
\begin{equation}
     \sqrt{n} \left( \hat{B}_{t_1,\dots,t_k} - B_{t_1,\dots,t_k} \right) = \frac{1}{\sqrt{n}}\sum_{i=1}^n \projection{t_1,\dots,t_k}{\data_i} + o_\mathbbmss{P}(1)\,,
\end{equation}
from which the required CLT follows.

By Slutsky theorem, the previous result holds also when $\hat{\Sigma}_{t_1,\dots,t_k} \pto \Sigma_{t_1,\dots,t_k}$ replaces $\Sigma_{t_1,\dots,t_k}$. 

\subsection{Theorem \ref{th:gp}}
According to Theorem 7.5 in \citet{billingsley2013convergence}, we need to show that:
\begin{itemize}
    \item The finite-dimensional projections of $\drfos$ converge to normal distributions;
    \item $\lim_{\delta\to0}\limsup_{n\to\infty}\PP{w(\drfos;\delta)\geq \varepsilon} = 0$, where $w(\drfos;\delta) = \sup_{|s-t|\leq\delta} |\drfos(s) - \drfos(t)|$ is the \textit{modulus of continuity}. 
\end{itemize}

The first requirement is satisfied by Lemma \ref{lem:pan}. Here we focus on the second requirement, which is essentially \textit{tightness}, \lorenzo{or \textit{stochastic equicontinuity}.}

First, by applying the triangle inequality on the modulus of continuity one obtains:
\begin{equation}
\label{eq:Mod_cont}
\begin{split}
    w(\drfos;\delta) &= \sup_{|s-t|\leq\delta} |\drfos(s) - \drfos(t)| \\
    &\leq \frac{1}{n} \sum_{i=1}^n \left(1 + \frac{A_i}{\hat{\pi}^{(1)}(X_i)}\right) \sup_{|s-t|\leq\delta} \left|\hat{\mu}^{(1)}(X_i;s) - \hat{\mu}^{(1)}(X_i;t)\right| \\
    &+ \frac{1}{n} \sum_{i=1}^n \left(1 + \frac{1 - A_i}{1 - \hat{\pi}^{(1)}(X_i)}\right) \sup_{|s-t|\leq\delta} \left|\hat{\mu}^{(0)}(X_i;s) - \hat{\mu}^{(0)}(X_i;t)\right| \\
    &+\frac{1}{n} \sum_{i=1}^n \frac{A_i - \hat{\pi}^{(1)}(X_i)}{\hat{\pi}^{(1)}(X_i)(1 - \hat{\pi}^{(1)}(X_i))} \sup_{|s-t|\leq\delta} \left| \mathcal{Y}_i(s) - \mathcal{Y}_i(t)\right| \\
    &= M(\drfos;\delta)\,.
\end{split}
\end{equation}

Notice that $w(\drfos;\delta)\geq0$. Therefore, by Markov inequality, it holds that:
\begin{equation}
\label{eq:markov}
    \PP{w(\drfos;\delta)\geq \varepsilon} \leq \frac{\EE{w(\drfos;\delta)}}{\varepsilon} \leq \frac{\EE{M(\drfos;\delta)}}{\varepsilon}\,,
\end{equation}
where the second inequality follows from Eq.~\ref{eq:Mod_cont}. \lorenzo{Crucially, by Assumption~\ref{ass:estimate}\textbf{c} and the use of cross-fitting, $\hat{\pi}^{(1)}(X_i)$ is bounded away from 0 and 1 with probability 1. This ensures the uniform boundedness in probability of the inverse weights, $(1+A_i/\hat{\pi}^{(1)}(X_i))$ and $(1+(1-A_i)/(1-\hat{\pi}^{(1)}(X_i)))$, which is necessary for the Markov inequality argument to hold.}
We can now take limits, and obtain:
\begin{equation}
    \lim_{\delta\to0}\limsup_{n\to\infty}\PP{w(\drfos;\delta)\geq \varepsilon} \leq \lim_{\delta\to0}\limsup_{n\to\infty}\frac{\EE{M(\drfos;\delta)}}{\varepsilon}\,.
\end{equation}

The three terms in $M(\drfos;\delta)$ can all be bounded using the assumptions on expected Hölder continuity. In fact, exploiting the assumption of no unmeasured confounding and assumption \ref{ass:estimate}\textbf{d}, the first term is
\begin{equation}
    \limsup_{n\to\infty}\frac{1}{n} \sum_{i=1}^n \EE{1 + \frac{A_i}{\hat{\pi}^{(1)}(X_i)}} \EE{\sup_{|s-t|\leq\delta} \left|\hat{\mu}^{(1)}(X_i;s) - \hat{\mu}^{(1)}(X_i;t)\right|} \overset{{\delta\to0}}{\rightarrow} 0\,.
\end{equation}
A similar argument holds for the second term. The third term is
\begin{equation}
\limsup_{n\to\infty}\frac{1}{n} \sum_{i=1}^n \EE{\frac{A_i - \hat{\pi}^{(1)}(X_i)}{\hat{\pi}^{(1)}(X_i)(1 - \hat{\pi}^{(1)}(X_i))}} \EE{\sup_{|s-t|\leq\delta} \left| \mathcal{Y}_i(s) - \mathcal{Y}_i(t)\right|} \overset{{\delta\to0}}{\rightarrow} 0\,,
\end{equation}
where we used assumption \ref{ass:estimate}\textbf{e}.

\lorenzo{Finally, the estimator $\drfos$ is defined as a linear combination of continuous functions ($\mathcal{Y}_i$ and the estimated continuous regression functions $\hat{\mu}^{(a)}$), ensuring that $\drfos$ is a continuous function and therefore a measurable element of $C(\mathcal{T})$.}

\section{Functional IPW}
\label{suppsec:fipw}
\lorenzo{Here, we operate in the space $C(\mathcal{T})$ equipped with the Borel $\sigma$-algebra generated by the sup-norm topology. Since the outcome $\mathcal{Y}$ is continuous, the pointwise operations (multiplication by the scalar weights and difference) preserve continuity, confirming that the map 
\begin{equation}
    (X, A, \mathcal{Y}) \mapsto \frac{A \mathcal{Y}}{\hat{\pi}^{(1)}(X)} - \frac{(1-A) \mathcal{Y}}{1-\hat{\pi}^{(1)}(X)}
\end{equation} 
is a measurable element of $C(\mathcal{T})$.}

The \textit{functional inverse probability weighting} (\texttt{IPW}) estimator is \lorenzo{then} defined as
\begin{equation}
    \hat{\beta}_{\texttt{IPW}} = \frac{1}{n} \sum_{i=1}^n \left[ \frac{A_i\mathcal{Y}_i}{\hat{\pi}^{(1)}(X_i)} - \frac{(1-A_i) \mathcal{Y}_i}{1 - \hat{\pi}^{(1)}(X_i)}\right]\,.
\end{equation}
We first show that the functional \texttt{IPW} is unbiased if the true propensity score is known.
\begin{proposition}
    Under Assumptions~\ref{ass:identify} (identifiability), if $\hat{\pi}^{(1)} = \pi^{(1)}$ a.s., one has
    \begin{equation}
        \EE{\hat{\beta}_{\texttt{IPW}}} = \beta\,.
    \end{equation}
\end{proposition}
\begin{proof}
    By properties of expectation and by the law of iterated expectations, we have
    \begin{equation}
        \begin{split}
            \EE{\hat{\beta}_{\texttt{IPW}}} &= \frac{1}{n} \sum_{i=1}^n \left[ \EE{\frac{A\mathcal{Y}}{\hat{\pi}^{(1)}(X)}} - \EE{\frac{(1-A) \mathcal{Y}}{1 - \hat{\pi}^{(1)}(X)}}\right] \\
            &=  \EE{\EE{\frac{A\mathcal{Y}}{\hat{\pi}^{(1)}(X)}\,\Big|\,X}} - \EE{\EE{\frac{(1-A) \mathcal{Y}}{1 - \hat{\pi}^{(1)}(X)}\,\Big|\,X}} \\
            &= \EE{\frac{\EE{A\,|\,X}}{\hat{\pi}^{(1)}(X)}\EE{\mathcal{Y}^{(1)}\,|\,X}} - \EE{\frac{\EE{1 - A \,|\,X}}{1 - \hat{\pi}^{(1)}(X)}\EE{\mathcal{Y}^{(0)}\,|\,X}} \\
            &= \EE{\mathcal{Y}^{(1)} - \mathcal{Y}^{(0)}} \\
            &= \beta\,.
        \end{split}
    \end{equation}
\end{proof}

The previous result implies that if $\hat{\pi}^{(1)}\pto\pi^{(1)}$, then the functional \texttt{IPW} will be consistent estimator of the FATE. 

We now show that the functional \texttt{IPW} asymptotically behaves as a Gaussian process if the true propensity score is known. First, we introduce the following assumptions. \\

\begin{assumption}[Inference for \texttt{IPW}]
\label{ass:estimateIPW}
    Assume that: 
    \begin{enumerate}[label=\textbf{\alph*.}]
        \item For any $t,s\in\mathcal{T}$, $\Sigma(s,t) = \text{cov}(\zeta(t),\zeta(s)) < \infty\,$, where $\zeta(t) = \frac{A\mathcal{Y}(t)}{\hat{\pi}^{(1)}(X)} - \frac{(1-A) \mathcal{Y}(t)}{1 - \hat{\pi}^{(1)}(X)}$.
        \item Given $\xi>0$, $\pi^{(a)}$ is bounded away from $\xi$ and $1-\xi$ with probability 1.
        \item For any $\delta>0$, the functional outcome satisfies
        \begin{equation}
            \EE{\sup_{|s-t|\leq \delta} \left|\mathcal{Y}(s)-\mathcal{Y}(t)\right|} \leq L \delta^\tau\,
        \end{equation}
        for some constants $L\geq0$ and $\tau>0$.
    \end{enumerate}
\end{assumption}
\bigskip

\begin{theorem}
\label{th:gp_ipw}
   Under Assumptions \ref{ass:identify} (identifiabiliy) and \ref{ass:estimateIPW} (inference), if $\hat{\pi}^{(1)} = \pi^{(1)}$ a.s., one has
    \begin{equation}
        \sqrt{n}\left(\hat{\beta}_{\texttt{IPW}} - \beta \right) \dto \mathcal{GP}(0, \Sigma)\,,
    \end{equation}
    where $\Sigma(s,t) = \text{cov}(\zeta(t),\zeta(s))\,$, $\zeta(t) = \frac{A\mathcal{Y}(t)}{\hat{\pi}^{(1)}(X)} - \frac{(1-A) \mathcal{Y}(t)}{1 - \hat{\pi}^{(1)}(X)}$. 
\end{theorem}
\begin{proof}
    The proof of this result follows again from Theorem 7.5 in \citet{billingsley2013convergence}. In particular, the standard Central Limit Theorem guarantees that scaled finite-dimensional projections of $\hat{\beta}_\texttt{IPW}$ converge to multivariate normal distributions. It remains to be shown that the condition on the modulus of continuity holds. This follows from the proof of Theorem~\ref{th:gp}, setting $\hat{\mu}^{(1)} = \hat{\mu}^{(0)} = 0$. 
\end{proof}

\begin{remark} 
\lorenzo{The \texttt{IPW} estimator can alternatively be interpreted as an expectation taken under a weighted law, effectively a change of measure via a Radon-Nikodým derivative between the observational and counterfactual laws \citep{bruns2025augmented, raykov2025kernel}. Our Gaussian Process limit for $\hat{\beta}_{\text{IPW}}$ (Theorem~\ref{th:gp_ipw}) confirms the expected functional Central Limit Theorem behavior for this weighted empirical process. Under the positivity condition (Assumption~\ref{ass:estimateIPW}\textbf{b}), the weights are bounded, and combined with the continuity of sample paths (Assumption~\ref{ass:estimateIPW}\textbf{c}), this is sufficient to ensure tightness in $C(\mathcal{T})$, aligning our result with the theoretical frameworks developed under the RN perspective.}
\end{remark}

\clearpage
\section{Further simulation details}
\label{suppsec:simulations}

 %Also, notice that the estimator proposed by \citet{ecker2024causal} is an outcome regression estimator with a function-on-scalar least squares specification of $\hat{\mu}$, that is:
% \begin{equation}
%     \hat{\mu}^{(a)}(X_i) = X_i \hat{\theta}^{(a)}_{\texttt{OLS}}\,,\quad\hat{\theta}^{(a)}_{\texttt{OLS}} = \left(X^{{(a)}^T} X^{(a)} \right)^{-1} X^{{(a)^T}} \mathcal{Y}^{(a)}\,.
% \end{equation}

% In high-dimensional settings where the number of observations $n$ is smaller than the number of predictors $p$, the previous estimator cannot be computed. Therefore, in our simulations we also fit a regularized model, namely \texttt{FAStEN} \citep{boschi2021highly, boschi2024fasten}% and \texttt{FunGCN} \citep{boschi2024fungcn}. Similarly, we model the propensity score using standard ML algorithms, namely logistic regression in low-dimensional scenarios and random forest classifier in high dimensions \citep{breiman2001random}.

%Finally, in high-dimensional scenarios, we impose sparsity by setting to 0 all but $|\mathcal{J}_\mu|=10$ regression coefficient curves, so that only $|\mathcal{J}_\mu|=10$ features remain active. The same sparsity is also induced in the scalar coefficients $W_j$, where only $|\mathcal{J}_\pi|=10$ variables are kept active. Again, notice that $\mathcal{J}_\mu \cap \mathcal{J}_\pi = \emptyset$. This structure allows us to study and compare the performance of the various estimators under misspecified scenarios, by progressively removing 2, 4, and 6 active features from the input data of either the outcome regression or propensity score model. 

% \newpage
\subsection{Additional simulation results}
\begin{figure}[h!]
    \centering
    \includegraphics[width=\linewidth]{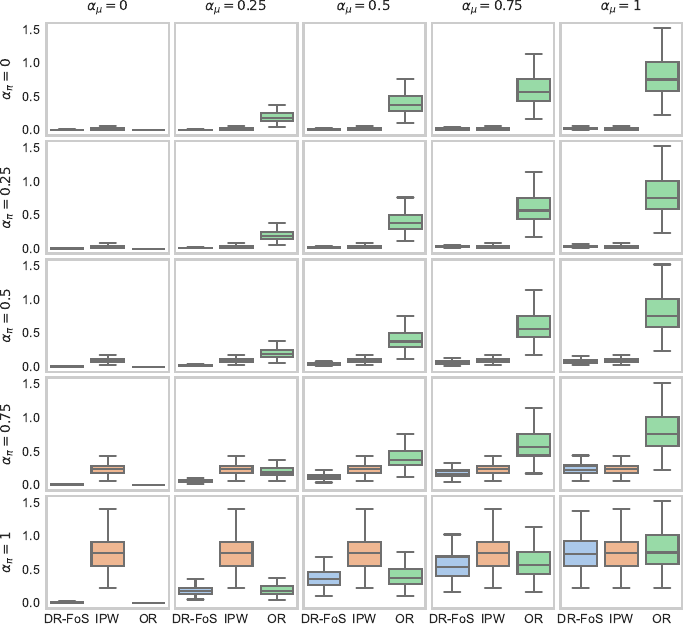}
    \caption{Results of simulation study run as described in Section \ref{sec:simulate}. Each panel displays estimation performance of \DRFOS{}, \texttt{IPW} and \texttt{OR} under different $\alpha_\mu$ and $\alpha_\pi$. For each box, the center line represents the median; the lower and upper hinges correspond to the first and third quartiles; the upper and lower whiskers span 1.5 times the interquartile range. \DRFOS{} consistently matches or outperforms the other estimators, demonstrating the strength of double robustness.}
    \label{fig:full_sim}
\end{figure}

\subsection{Simulations with non-smooth trajectories}
\label{sec:jump}
\lorenzo{Here, we repeat the simulation study described in Section \ref{sec:simulate}, adding discontinuitites to the regression functions. We expect that, in the presence of sharp local irregularities (as jumps or cusps), simultaneous coverage bands can become overly large. Instead of using smooth regression functions, we generate the true regression functions according to the model
\begin{equation}
    \mu^{(a)}(t) = a\Tilde{\beta}(t) + \rho(t) + 3\one^{(a)}(t)\,,\quad a\in\{0,1\}\,,
\end{equation}
where the $\Tilde{\beta}$ and the baseline function $\rho$ are drawn from a $0$ mean Gaussian process with a Matern covariance function as before, while $\one^{(a)}(t)$ is defined as:
\begin{equation}
   \one^{(a)}(t) = \begin{cases}
        1 & a=1, 0\leq t\leq 1/3\\
        1 & a=0, 2/3\leq t\leq 1\\\
        0 & \text{otherwise}\,.
    \end{cases}
\end{equation} 
The true FATE in this simulation is $\beta(t) = \Tilde{\beta}(t) + 3\one^{(a)}(t)$.}
\begin{figure}
    \centering
    \includegraphics[width=\linewidth]{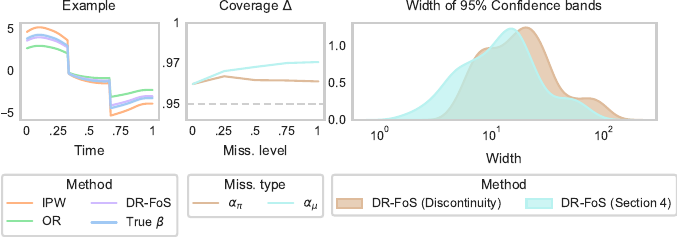}
    \caption{\lorenzo{Simulation results based on the data generating process described in Section~\ref{sec:simulate} with the discontinuity in the regression functions as described in Appendix Section \ref{sec:jump}. The first panel (leftmost) shows an example of true FATE $\beta$ (blue) as generated by our simulation scheme, and the estimates provided by \texttt{IPW} (orange), \texttt{OR} (green) and \DRFOS{} (purple). In this example we set $\alpha_\mu=0.25$ and $\alpha_\pi=0.75$. The second panel shows the average coverage levels achieved under misspecifications of the regression functions (light blue, $\alpha_\pi=0$) and of the propensity score (brown, $\alpha_\mu=0$). The grey dashed line represents the nominal $95\%$ coverage level. The third panel (rightmost) compares the average widths of the 95\% simultaneous confidence bands from two simulations, displayed on a logarithmic axis. Specifically, it plots the distribution of the average band width from the simulation with regression function discontinuities against the average band width from the original (continuous) simulation (described in Section \ref{sec:simulate}). The plot clearly demonstrates that the introduction of discontinuities poses a significant challenge. To maintain the 95\% coverage guarantee under these conditions, \DRFOS{} robustly manages the issue by enlarging the width of the confidence band.}}
    \label{fig:jump}
\end{figure}

\lorenzo{Figure \ref{fig:jump} confirms our intuition: to maintain valid 95\% coverage guarantee under this new setting with discontinuities, \DRFOS{} enlarges the width of the confidence bands.}

\subsection{Simulations with an alternative, explicit formulation}
\label{suppsec:expl_sim}

We complement the main simulation study by presenting an additional scenario based on an explicit function-on-scalar data generating process with linear structure, mimicking the setting presented by \citet{liu2024double}. This simulation study closely mirrors classical function-on-scalar regression setups and serves to validate the robustness of the proposed estimator \DRFOS{} against model misspecification.

For each unit $i = 1, \dots, n$, we generate covariates $X_i \in \RR^p$, with $p = 5$ and $X_i \sim \text{Unif}(-1,1)^p$; treatment assignments $A_i \sim \text{Ber}(\pi(X_i))$, where the propensity score takes the logistic form $\pi(X_i) = (1 + \exp(-X_i^T \eta + \epsilon_i))^{-1}$, the coefficients $\eta \in \RR^p$ are drawn once from $\Normal{0}{I_p}$, and random noise is sampled according to $\epsilon_i\sim\Normal{0}{1}$; potential outcomes $\mathcal{Y}_i^{(0)} = \sum_{j=1}^p X_i \theta^{(0)}_j + \varepsilon_i$ and $\mathcal{Y}_i^{(1)} = \sum_{j=1}^p X_i \theta^{(1)}_j + D + \varepsilon_i$, where each $\theta^{(0)}_j$ and $\theta^{(1)}_j$ are drawn from a Gaussian process with a Matern covariance function as in Equation~\ref{eq:matern}, with parameters $l=0.25$, $\nu=3.5$, and choose $\eta=1$, $D$ is a constant function set at $1$, and the functional error $\varepsilon_i$ is drawn pointwise from a standard normal distribution; the observed outcome is $\mathcal{Y}_i = A_i \mathcal{Y}_i^{(1)} + (1-A_i) \mathcal{Y}_i^{(0)}$; the true functional average treatment effect (FATE) is $\beta = \EE{\mathcal{Y}^{(1)}_i - \mathcal{Y}^{(0)}_i}$.

We estimate $\beta$ using the outcome regression model in Equation~\ref{eq:or}, the \texttt{IPW} estimator in Equation~\ref{eq:ipw}, and \DRFOS{} as defined in Equation~\ref{eq:drfos}. All estimators employ a logistic regression model for learning $\hat{\pi}^{(a)}$, while we exploit both standard function-on-scalar linear models \citep{kokoszka2017introduction} and a state-of-the-art machine learning method, called \texttt{FunGCN} \citep{boschi2024fungcn}, to learn $\hat{\mu}^{(a)}$. To fit \texttt{FunGCN}, we set the following hyperparameters: $forecast\_ratio = 0$, $pruning = 0.7$, $k\_gcn = 10$, $lr = 5e-5$, $max\_selected = 5$, $k\_graph = 3$, $nhid = [32, 32]$, $epochs = 50$, $batch\_size = 1$, $dropout = 0$, $kernel\_size = 0$, $patience = 5$, $min\_delta = 0$, $val\_size = 0$, $test\_size = 0$. We refer to \citet{boschi2024fungcn} for an in-depth description of these hyperparameters.

To assess robustness, we introduce a misspecified feature set $\Tilde{X}_i\in \RR^3$, derived from the original covariates via nonlinear and interaction terms:
\begin{equation}
    \Tilde{X}_i = \begin{bmatrix}
\sin(X_{i1}) \\
(X_{i2} + X_{i3})^2 \\
\log(1 + |X_{i4}|)
\end{bmatrix}\,.
\end{equation}
Under well-specified scenarios, both the outcome regression and propensity score models are fitted using $X_i$. Under misspecification, nuisance models are fitted using $\Tilde{X}_i$ instead of $X_i$. All other aspects of the simulation, including the true data-generating process, remain unchanged.

\begin{figure}
    \centering
    \includegraphics[width=\linewidth]{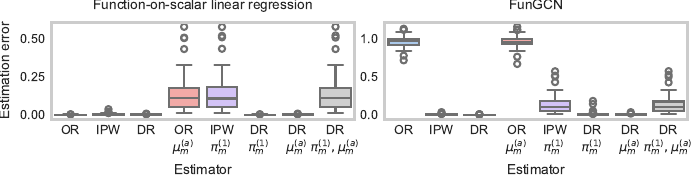}
    \caption{Results of the simulation study described in Appendix Section~\ref{suppsec:expl_sim}. Each panel shows the estimation performance of \DRFOS{}, \texttt{IPW}, and \texttt{OR} under both well-specified and misspecified settings. Misspecification is introduced by fitting nuisance functions using transformed covariates $\Tilde{X}_i$, $i = 1, \dots, n$. Misspecified outcome regressions are denoted by $\mu^{(a)}_m$, and misspecified propensity scores by $\pi^{(1)}_m$. For example, \texttt{OR} $\mu^{(a)}_m$ refers to the outcome regression estimator fitted with a misspecified regression function, DR $\pi^{(1)}_m$ denotes the \DRFOS{} estimator fitted with a misspecified propensity score but well-specified regression functions, while DR $\pi^{(1)}_m,\mu^{(a)}_m$ denotes the \DRFOS{} estimator constructed using both a misspecified propensity score and misspecified regression functions.
    For each box, the center line represents the median; the lower and upper hinges correspond to the first and third quartiles; the upper and lower whiskers span 1.5 times the interquartile range. \DRFOS{} consistently matches or outperforms the other estimators, demonstrating the strength of double robustness.}
    \label{fig:expl_sim}
\end{figure}

We repeat each experiment across 50 independent random seeds. The results are summarized in Figure~\ref{fig:expl_sim}. Under well-specified conditions, the outcome regression (\texttt{OR}) estimator based on a linear function-on-scalar model exhibits excellent performance, with consistently low estimation error across replications. In contrast, the \texttt{OR} estimator using \texttt{FunGCN} displays slightly higher variability due to the model’s tendency to oversmooth functional predictions. The \texttt{IPW} estimator also performs strongly in the well-specified setting, and as predicted by theory, \DRFOS{} matches or exceeds this performance by combining both nuisance components. In particular, when \texttt{FunGCN} is used for outcome regression, \DRFOS{} clearly benefits from the \texttt{IPW} correction, compensating for the oversmoothing bias. In misspecified scenarios, both the \texttt{OR} and \texttt{IPW} estimators degrade significantly, confirming their sensitivity to model misspecification. In contrast, \DRFOS{} maintains superior performance when either the regression or the propensity score model is correctly specified, and provides comparable performance even when both are misspecified. Notably, in the \texttt{FunGCN} setting, \DRFOS{} consistently outperforms the misspecified \texttt{OR} estimator, once again highlighting the value of the \texttt{IPW} correction in mitigating outcome model bias.

% We generate the response additively according to the model
% \begin{equation}
%     \mathcal{Y}_i = A_i\beta + \sum_{j=1}^p X_{ij}\theta_j + \varepsilon_i\,,
% \end{equation}
% where each feature $X_{j}$ is drawn from an independent Standard Normal distribution, i.e.~$X_{ij}\sim\Normal{0}{1}$; the treatment indicator $A_i$ is drawn from a Bernoulli distribution with parameter $\Tilde{p}_i = \text{expit}\,(XW)$, where $W_j\sim\Normal{0}{1}$, $j\in\{1,\dots,p\}$; the functional coefficients $\beta$ and $\theta_j$, $j\in\{1,\dots,p\}$, are drawn from a $0$ mean Gaussian process with a Matern covariance function using $l=0.25$, $\nu=3.5$, $\eta^2=1$. Similarly, the error $\varepsilon$ is drawn from the same Matern covariance process, with $l=0.25$ and $\nu=2.5$. The value of $\eta^2$ is set as $\VV{A_i\beta +\sum_{j=1}^p X_{ij}\theta_j} / 10$, where with a slight abuse of notation $\VV{A_i\beta + \sum_{j=1}^p X_{ij}\theta_j}$ indicates the global (pooled over time) variance of the signal $A_i\beta + \sum_{j=1}^p X_{ij}\theta_j$. This roughly ensures that the pooled signal-to-noise ratio is approximately 1.

\section{Additional real-data analysis}
\label{suppsec:more_real}

In addition to the analysis presented in the main text, here we provide a further robustness check, where we use a state-of-the-art machine learning method for functional data called \texttt{FunGCN} \citep{boschi2024fungcn}, to fit the regression model $\hat\mu^{(a)}$. \lorenzo{\texttt{FunGCN}, which leverages a Graph Convolutional Network architecture, serves as a non-linear function-on-scalar regression tool. It is specifically designed to handle multimodal data and learn complex dependencies between functional outcomes and scalar/vector covariates, even on irregular observation grids. This approach differs from earlier work like \texttt{funcNN} \citep{thind2020funcnn} by explicitly using graph-based message passing to capture dependencies across the functional domain and integrate scalar covariates efficiently into the network.} To fit \texttt{FunGCN}, we set the following hyperparameters: $forecast\_ratio = 0$, $pruning = 0.7$, $k\_gcn = 10$, $lr = 5e-5$, $max\_selected = 5$, $k\_graph = 3$, $nhid = [32, 32]$, $epochs = 50$, $batch\_size = 1$, $dropout = 0$, $kernel\_size = 0$, $patience = 5$, $min\_delta = 0$, $val\_size = 0$, $test\_size = 0$. We refer to \citet{boschi2024fungcn} for an in-depth description of these hyperparameters.

\begin{figure}
    \centering
    \includegraphics[width=\linewidth]{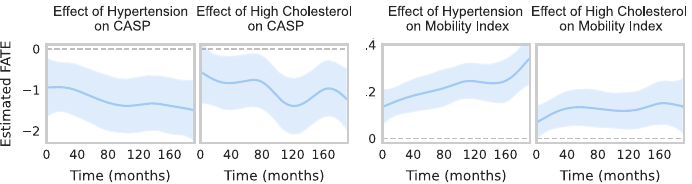}
    \caption{SHARE application results using \texttt{FunGCN} as model for the regression function. Each panel displays a different estimated causal effect. Blue continuous lines correspond to \DRFOS{} estimates; blue bands are 95\% asymptotic simultaneous confidence bands obtained by repeatedly sampling from the Gaussian process; grey dotted horizontal lines correspond to 0.}
    \label{fig:app_gcn}
\end{figure}

The estimation strategy is the same as in the main text. For each combination of chronic disease and functional outcome, we fit \DRFOS{} using \texttt{FunGCN} for the regression function $\hat{\mu}^{(a)}$, and logistic regression for the propensity score $\hat{\pi}^{(a)}$. We employ cross-fitting with 5 balanced folds. Results are shown in Figure~\ref{fig:app_gcn} and fully align with the findings in the main text, supporting their robustness.

\end{document}